\documentclass[review]{elsarticle}
\usepackage{graphicx}
\usepackage{epstopdf}
\usepackage{amsmath}
\usepackage{amssymb}
\usepackage[linewidth=.8pt]{mdframed}
\usepackage{tcolorbox}
\usepackage{bm}
\usepackage{nicefrac}
\usepackage{booktabs}
\usepackage[tableposition=top]{caption}
\usepackage{tikz}
\usepackage[utf8]{inputenc}
\usepackage[english]{babel}
\usepackage{pgfplotstable}
\usepackage{array}
\usepackage{colortbl}
\usepackage{booktabs}
\usepackage{eurosym}
\usepackage{amsmath}
\usepackage{pgfplotstable}
\usepackage{xcolor}
\usepackage{pgfplots}
\pgfplotsset{compat=newest}
\usepgfplotslibrary{groupplots}

\usepackage{array}
\usepackage{multirow}
\usepackage{threeparttable}
\usepackage{makecell}
\usepackage[procnumbered,ruled,vlined,linesnumbered]{algorithm2e}
\usepackage{siunitx}
\usepackage{stfloats}
\usepackage{graphicx}
\usepackage{subfigure}
\usepackage{hyperref}
\usepackage{textcomp}
\usepackage{array}
\usepackage{epstopdf}
\usepackage{framed}
\usepackage{cleveref}
\usepackage[normalem]{ulem}
\usepackage{xcolor}
\usepackage{amsfonts}
\usepackage{booktabs}
\usepackage{siunitx}

\newcommand\hl{\bgroup\markoverwith
  {\textcolor{yellow}{\rule[-.5ex]{2pt}{2.5ex}}}\ULon}

\newtheorem{problem}{Problem}

\newtheorem{theorem}{Theorem}[section]

\newtheorem{definition}[theorem]{Definition}
\newtheorem{remark}[theorem]{Remark}

\newcommand{\PreserveBackslash}[1]{\let\temp=\\#1\let\\=\temp}
\newcolumntype{C}[1]{>{\PreserveBackslash\centering}p{#1}}
\newcolumntype{R}[1]{>{\PreserveBackslash\raggedleft}p{#1}}
\newcolumntype{L}[1]{>{\PreserveBackslash\raggedright}p{#1}}

\newenvironment{fminipage}%
{\begin{Sbox}\begin{minipage}}%
		{\end{minipage}\end{Sbox}\fbox{\TheSbox}}

\def\norm#1{\left\| #1 \right\|}
\def\smallnorm#1{\| #1 \|}

\def\calG{\mathcal{G}}

\def\calN{\mathcal{N}}

\newfont{\nset}{msbm10}

\def\norm#1{\left\| #1 \right\|}

\newcommand{\removelatexerror}{\let\@latex@error\@gobble}

\newcommand\LL{\bm{\mathit{L}}}

\newtheorem{thm}{Theorem}
\newtheorem{lem}[thm]{Lemma}
\newdefinition{rmk}{Remark}
\newproof{pf}{Proof}
\newproof{pot}{Proof of Theorem \ref{thm2}}

\newcommand\aaa{\boldsymbol{\mathit{a}}}

\newcommand\bb{\boldsymbol{\mathit{b}}}
\newcommand\cc{\boldsymbol{\mathit{c}}}

\newcommand\ee{\boldsymbol{\mathit{e}}}

\newcommand\ww{\boldsymbol{\mathit{w}}}
\newcommand\uu{\boldsymbol{\mathit{u}}}

\newcommand\AAA{\boldsymbol{\mathit{A}}}
\newcommand\BB{\boldsymbol{\mathit{B}}}

\newcommand\CC{\boldsymbol{\mathit{C}}}

\newcommand\DD{\boldsymbol{\mathit{D}}}

\newcommand\EE{\boldsymbol{\mathit{E}}}
\newcommand\PP{\boldsymbol{\mathit{P}}}
\newcommand\MM{\boldsymbol{\mathit{M}}}

\newcommand\II{\boldsymbol{\mathit{I}}}
\newcommand\vvv{\boldsymbol{\mathit{v}}}

\DontPrintSemicolon
\SetKw{KwAnd}{and}
\SetFuncSty{textsc}
\SetKwInOut{Input}{Input\ \ \ \ }
\SetKwInOut{Output}{Output}

\usepackage{tabularx}
\usepackage{stfloats}

\journal{Theoretical Computer Science}

\begin{document}
\begin{frontmatter}
\title{Optimization on the smallest eigenvalue of grounded Laplacian matrix via edge addition}
%\title{This is a specimen title\tnoteref{t1,t2}}
%\tnotetext[t1]{11}
%\tnotetext[t2]{11}
\author[1]{Xiaotian Zhou}
\ead{22110240080@m.fudan.edu.cn}
\author[1]{Haoxin Sun}
\ead{21210240097@m.fudan.edu.cn}
\author[2]{Wei Li\corref{cor1}}
\ead{fd\_liwei@fudan.edu.cn}
\author[1,3,4]{Zhongzhi Zhang\corref{cor1}}
\ead{zhangzz@fudan.edu.cn}
%\fnref{fn1,fn3}
\cortext[cor1]{Corresponding author}
%\fntext[fn1]{This is the first author footnote.}
%\fntext[fn2]{11}
%\fntext[fn3]{Yet another author footnote.}
\affiliation[1]{organization={Shanghai Key Laboratory of Intelligent Information Processing, School of Computer Science, Fudan University},
%addressline={Fudan University},
postcode={200433},
city={Shanghai},
country={China}}
\affiliation[2]{organization={Academy for Engineering and Technology, Fudan University},
%addressline={Fudan University},
postcode={200433},
city={Shanghai},
country={China}}
\affiliation[3]{organization={Shanghai Engineering Research Institute of Blockchains, Fudan University},
%addressline={Fudan University},
postcode={200433},
city={Shanghai},
country={China}}
\affiliation[4]{organization={Research Institute of Intelligent Complex Systems, Fudan University},
%addressline={Fudan University},
postcode={200433},
city={Shanghai},
country={China}}

%\thanks{This work was supported by the National Natural Science Foundation of China (Nos. 61872093 and U20B2051),  Shanghai Municipal Science and Technology Major Project  (Nos.  2018SHZDZX01 and 2021SHZDZX03),  ZJ Lab, and Shanghai Center for Brain Science and Brain-Inspired Technology. Run Wang was also supported by Fudan’s Undergraduate Research Opportunities Program (FDUROP) under Grant No. 2195200241021. (Run Wang and Xiaotian Zhou contribute equally. Corresponding authors:~Wei~Li and~Zhongzhi~Zhang)}
%\thanks{Run Wang, Xiaotian Zhou  and Zhongzhi Zhang are with the Shanghai Key Laboratory of Intelligent Information Processing, School of Computer Science, Fudan University, Shanghai 200433, China. Wei Li is with the Academy for Engineering and Technology, Fudan University, Shanghai, 200433, China. Zhongzhi~Zhang is also with the Shanghai Engineering Research Institute of Blockchains,  Fudan University, Shanghai 200433, China; and Research Institute of Intelligent Complex Systems, Fudan University, Shanghai 200433, China. (e-mail: runwang18@fudan.edu.cn; 20210240043@fudan.edu.cn; fd\_liwei@fudan.edu.cn; zhangzz@fudan.edu.cn).}

\begin{abstract}
The grounded Laplacian matrix $\LL_{-S}$ of a graph $\calG=(V,E)$ with $n=|V|$ nodes and $m=|E|$ edges is a $(n-s)\times (n-s)$ submatrix of its Laplacian matrix $\LL$, obtained from $\LL$ by deleting rows and columns corresponding to $s=|S| \ll n $ ground nodes forming set $S\subset V$. The smallest eigenvalue of $\LL_{-S}$ plays an important role in various practical scenarios, such as characterizing the convergence rate of leader-follower opinion dynamics, with a larger eigenvalue indicating faster convergence of opinion. In this paper, we study the problem of adding $k \ll n$ edges among all the nonexistent edges forming the candidate edge set $Q = (V\times V)\backslash E$, in order to maximize the smallest eigenvalue of the grounded Laplacian matrix. We show that the objective function of the combinatorial optimization problem is monotone but non-submodular. To solve the problem, we first simplify the problem by restricting the candidate edge set $Q$ to be $(S\times (V\backslash S))\backslash E$, and prove that it has the same optimal solution as the original problem, although the size of set $Q$ is reduced from $O(n^2)$ to $O(n)$. Then, we propose two greedy approximation algorithms. One is a simple greedy algorithm with an approximation ratio $(1-e^{-\alpha\gamma})/\alpha$ and time complexity $O(kn^4)$, where  $\gamma$ and $\alpha$ are, respectively, submodularity ratio and curvature, whose bounds are provided for some particular cases. The other is a fast greedy algorithm without approximation guarantee, which has a running time $\tilde{O}(km)$, where $\tilde{O}(\cdot)$ suppresses the ${\rm poly} (\log n)$ factors. Numerous experiments on various real networks are performed to validate the superiority of our algorithms, in terms of effectiveness and efficiency.
%The grounded Laplacian matrix $\LL_{-S}$ of a graph $\calG=(V,E)$ with $n=|V|$ nodes and $m=|E|$ edges is a $(n-s)\times (n-s)$ submatrix of its Laplacian matrix $\LL$, obtained from $\LL$ by deleting rows and columns corresponding to $s=|S|\ll n $ ground nodes forming set $S\subset V$. The smallest eigenvalue of $\LL_{-S}$ plays important roles in various practical scenarios, with a larger eigenvalue indicating better performance. In this paper, we study the problem of adding $k\ll n$ edges among all nonexistent edges forming the candidate edge set $Q = (V\times V)\backslash E$, to maximize the smallest eigenvalue of the grounded Laplacian matrix. We show that the objective function of the problem is monotone but non-submodular. To solve the problem, we first simplify the problem by reducing the candidate edge set. Then, we propose two greedy algorithms. One is a simple greedy algorithm with an approximation ratio $(1-e^{-\alpha\gamma})/\alpha$ and time complexity $O(kn^4)$, where $\gamma$ and $\alpha$ are, respectively, submodularity ratio and curvature. The other is a fast greedy algorithm without approximation guarantee, which has a running time $\tilde{O}(km)$, where $\tilde{O}(\cdot)$ suppresses the ${\rm poly}(\log n)$ factors. We perform numerous experiments on various networks to validate the effectiveness and efficiency of our algorithms.
\end{abstract}

\iffalse
\begin{graphicalabstract}
%\includegraphics{grabs}
\end{graphicalabstract}

%%Research highlights
\begin{highlights}
\item Research highlight 1
\item Research highlight 2
\end{highlights}
\fi

\begin{keyword}
Grounded Laplacian \sep spectral property \sep graph mining \sep linear algorithm \sep matrix perturbation \sep partial derivative
%\PACS 0000 \sep 1111
%% MSC codes here, in the form: \MSC code \sep code
%% or \MSC[2008] code \sep code (2000 is the default)
%\MSC 0000 \sep 1111
\end{keyword}
\end{frontmatter}
% make the title area
%\newpageafter{title}
%\newpageafter{author}
%\newpageafter{abstract}

\section{Introduction}

The eigenvalues of different matrices associated with a network encode rich information of various structural properties and dynamical processes taking place on the network~\cite{Ne03}. It has been shown that the largest eigenvalue of the adjacency matrix characterizes the thresholds for the susceptible-infectious-susceptible epidemic dynamics~\cite{WaChWaFa03,ChWaWaLeFa08,VaOmKo08} and bond percolation~\cite{BoBoChRi10} on a graph. The number of distinct eigenvalues of the adjacency matrix provides an upper bound for the diameter of a connected graph~\cite{Ho05}. And the sum of the powers of each eigenvalue measures the structural robustness in complex networks~\cite{WuBaTaDe11}. In the context of Laplacian matrix, its smallest and largest nonzero eigenvalues are, respectively, related to the convergence time and delay robustness of the consensus problem~\cite{OlMu04}, while the ratio of the largest eigenvalue and the smallest nonzero eigenvalue represents the synchronizability of the graph~\cite{BaPe02}, with smaller ratio corresponding to better synchronizability. Moreover, it has been demonstrated that the product of all the nonzero eigenvalues of the Laplacian matrix determines the number of spanning trees~\cite{LiPaYiZh20}, and that the reciprocal sum of these eigenvalues determines the sum of resistance distances~\cite{KlRa93,LiZh18,YiZhPa20,YiYaZhZhPa22} and the sum of hitting times~\cite{Te91,ChRaRuSm89,ShZh19} over all node pairs.

Apart from the adjacency matrix and Laplacian matrix, the eigenvalues of the grounded Laplacian matrix of a graph also characterize various dynamical processes or systems defined on the graph~\cite{RaJiMeEg09,PaBa10, PiShFiSu18, LiPeShYiZh19,LiXuLuChZe21,ToRo22}. For a graph $\calG=(V,E)$ with Laplacian matrix $\LL$, the grounded Laplacian matrix $\LL_{-S}$ is a principal submatrix of $\LL$, which is obtained from $\LL$ by removing rows and columns corresponding to nodes in a given set $S\subseteq V$~\cite{Mi93}. The $s=|S|$ nodes in set $S$ are called grounded nodes, which have different meanings in different dynamical processes. In leader-follower opinion dynamics~\cite{RaJiMeEg09}, $S$ denotes the set of leader nodes, while in the pinning control system~\cite{LiXuLuChZe21}, $S$ represents the set of pinned nodes. For an undirected connected graph, matrix $\LL_{-S}$ are positive definite, and its eigenvalues, especially the smallest one denoted by $\lambda(\LL_{-S})$, are related to diverse dynamical processes. For example, the sum of the reciprocal of eigenvalues of $\LL_{-S}$ measures the robustness of a leader-follower system with follower nodes subject to noise~\cite{PaBa10} or the centrality of a node group~\cite{LiPeShYiZh19}. The smallest eigenvalue $\lambda(\LL_{-S})$ quantifies the effectiveness of a pinning scheme for pinning control~\cite{LiXuLuChZe21}, as well as the convergence rate of leader-follower dynamical systems~\cite{RaJiMeEg09}, with larger $\lambda(\LL_{-S})$ corresponding to better performance of the systems.

Since the smallest eigenvalue $\lambda(\LL_{-S})$ of matrix $\LL_{-S}$ encodes the performance of various dynamical systems, with larger $\lambda(\LL_{-S})$ indicating better performance, a concerted effort has been devoted to optimizing $\lambda(\LL_{-S})$. Recently, the problem of optimizing $\lambda(\LL_{-S})$ has been studied by selecting a fixed number of grounded nodes~\cite{LiXuLuChZe21,WaLiXuLi19,ClAlBuPo12,ClBuPo12,ClHoBuPo18,ZhTa20}. However, the problem to optimize $\lambda(\LL_{-S})$ by edge operation has been less studied, despite the fact that edge operation has been frequently used in various application scenarios. As a less invasive modification of network structure, the operation of adding edges in a graph has good explanations in application settings, which is equivalent to making friends in social networks such as Twitter or establishing physical lines in power grids. In this paper, we propose and study the following optimization problem by adding edges: For a given connected undirected unweighted $\calG=(V,E)$ with $n$ nodes and $m\ll n^2$ edges, grounded node set $S\subset V$ of $s$ nodes with $s=|S|\ll n$, a positive integer $k\ll n$, how to add $k$ nonexistent edges  from a candidate edge set $Q=(V \times V) \backslash E$ to graph $\calG$, so that the smallest eigenvalue of the grounded Laplacian matrix for the resulting graph is maximized.

%Despite optimization on the smallest eigenvalue of the grounded Laplacian matrix has been heavily studied in recent years based on the selection of grounded nodes~\cite{LiXuLuChZe21,WaLiXuLi19,ClAlBuPo12,ClBuPo12,ClHoBuPo18,ZhTa20}, the operation on the edge level is rarely studied, and to the best of our knowledge, only~\cite{ShSoSaJo14} considered the problem of adding edges to maximize the smallest eigenvalue of a grounded Laplacian matrix. So we address the smallest eigenvalue maximization problem by adding edges: For a given connected undirected unweighted network $\calG=(V,E)$ with $n$ nodes and $m$ edges, and a $s$-sized grounded node set $S\subset V$, $s=|S|\ll n$, how to add $k\ll n$ new nonexistent edges from a candidate edge set $Q=(V \times V) \backslash E$ to the graph $\calG$, so that the smallest eigenvalue of the grounded Laplacian matrix is maximized.

The main reason for studying the above problem lies in its motivating applications in practical situations, such as leader-follower opinion dynamics and pinning control systems. For leader-follower opinion dynamics, the equilibrium behavior of opinions is relevant only if opinions converge in a reasonably small time~\cite{El93}. The operation of adding an edge indicates building friendship between the two individuals linked by the edge, which can lead to faster convergence of followers’ opinion to that of leaders. Regarding pinning control systems, adding edges such as establishing physical lines in a power grid can improve the effectiveness of pinning control without increasing pinned nodes. Our formulated problem is combinatorial in nature, we thus aim to provide a suboptimal solution. Our main contributions are summarized as follows.
\begin{itemize}
  \item We show that the objective function is non-submodular, although it is monotonically increasing.
  %\item We simplify the problem by reducing the candidate edge set $Q=(V \times V)\backslash E$ that contains all nonexistent edges to the edge set $Q=(S \times (V\backslash S))\backslash E$ that only contains nonexistent edges linking one grounded node and one non-grounded node. We also prove the optimal solution of the simplified problem is an optimal solution of the original problem.
 \item We simplify the problem by reducing the candidate edge set $Q=(V \times V)\backslash E$ with $O(n^2)$ edges to its proper set $ (S \times (V\backslash S))\backslash E$ with $O(n)$ edges, and prove that the simplified problem has the same optimal solution as the original problem.
 \item We propose a simple greedy algorithm with $O(kn^4)$ time complexity and an approximation ratio $(1-e^{\alpha\gamma})/\alpha$, where $\gamma$ and $\alpha$ are, respectively, the submodularity ratio and curvature of the objective function. We provide a lower bound for $\gamma$ and an upper bound for $\alpha$ for some particular cases.
  \item We propose a fast approximation  algorithm with $\tilde{O}(km)$ time complexity, where $\tilde{O}(\cdot)$ notation suppresses the ${\rm poly} (\log n)$ factors.
  \item We perform various experiments on real networks, which show that both of our algorithms are effective, outperforming alternative baselines, and that our fast algorithm  is scalable to large graphs with over one million nodes.
\end{itemize}

\section{Preliminaries}

In this section, we introduce some notations,  concepts of graphs and some related matrices, and concepts about set functions.

\subsection{Notations}

In this paper, we use normal lowercase letters like $a, b, c$ to denote scalars in $\mathbb{R}$,  normal uppercase letters like $A, B, C$ to represent sets,  bold lowercase letters like $\aaa, \bb, \cc$ to denote vectors, and  bold uppercase letters like $\AAA, \BB, \CC$ to denote matrices. We write $\aaa^\top$ and $\AAA^\top$to represent the transposes of vector $\aaa$ and matrix $\AAA$, respectively. We use $\aaa_i$  to denote the $i$-th element of vector $\aaa$, and notation $\AAA_{i,j}$ to denote entry $(i,j)$ of  matrix $\AAA$. Let $\II$ be the identify matrix of approximate dimension.
Let  $\EE_{i,j}$ denote the matrix with the element at   row $i$ and column  $j$   being 1 and other elements being 0. For a positive definite matrix with the minimum eigenvalue $\lambda$ and  corresponding unit eigenvector $\uu$, we call $(\lambda,\uu)$ as its smallest eigen-pair. We use $\AAA_{-S}$ to denote the submatrix of matrix $\AAA$ obtained from $\AAA$ by deleting those rows and columns corresponding to nodes in  set $S$, where if $S=\{i\}$, we use $\AAA_{-i}$ to represent $\AAA_{-\{i\}}$ for simplicity.

%For a $n\times n$ matrix $\AAAa$, $1\leq i \leq n$, we use $\AAAa_{-i}$ to denote the submatrix of $\AAAa$ by deleting $i$-th row and column of $\AAAa$. We use $\ee_i$ to denote the vector with $i$-th element being 1 while others being 0, and $\EE_{i,j}$ to denote the matrix with the element at $i$-th row and $j$-th column being 1 while others being 0.

\subsection{Graph and grounded Laplacian matrix}

Let $\calG=(V,E)$  denote a connected undirected unweighted graph with $|V|=n$ nodes and $|E|=m$ edges. The set $\calN_i$ of neighbors of node $i \in V$ is defined by $\calN_i=\{j \in V|(i,j)\in E \}$. Then the degree of node $i$ is $d_i=|\calN_i|$. The adjacency matrix $\AAA$ of graph $\calG$ is a $n \times n$ matrix,  whose entry $\AAA_{i,j}=\AAA_{j,i}=1$ if there is an edge $e=(i,j)\in E$ linking nodes $i$ and $j$, $\AAA_{i,j}=\AAA_{j,i}=0$ otherwise. The degree matrix $\DD$ of graph $\calG$ is   $\text{diag}(d_1,\dots,d_n)$  with the $i$-th  diagonal element being $d_i$. The Laplacian matrix of graph $\calG$ is $\LL=\DD-\AAA$, which is a symmetric positive semi-definite matrix with a unique smallest eigenvalue 0.

%\subsection{Grounded Laplacians}

For  graph $\calG=(V,E)$, a grounded Laplacian matrix $\LL_{-S}$ induced by set $S \subset V$ with $s=|S|$ nodes is an $(n-s)\times (n-s)$ principle submatrix of $\LL$ obtained by removing rows and columns of Laplacian matrix $\LL$ corresponding to those nodes in $S$, which are called grounded nodes. We use  $\calG[V \backslash S]$ to represent the subgraph of $\calG$, which is formed from  $\calG$ by deleting those nodes in set $S$ and the edges incident to nodes in $S$. The grounded Laplacian matrix $\LL_{-S}$ is a symmetric diagonally-dominant M-matrix (SDDM). By Perron–Frobenius theorem~\cite{Ma00}, the smallest eigenvalue  $\lambda(\LL_{-S})$ of matrix $\LL_{-S}$ is positive, and the components of  the eigenvector associated with  $\lambda(\LL_{-S})$ are positive (non-negative) if $\calG[V\backslash S]$ is connected (disconnected). Lemma~\ref{lem:lmd} provides the relation between the smallest eigenvalue $\lambda(\LL_{-S})$ and the degree of any node $i\in V\backslash S$.
%More precisely, both smallest eigenvalue and components in smallest eigenvalue are positive iff the subgraph induced by $V \backslash S$ is connected.

\begin{lem}\label{lem:lmd}
Given a graph $\calG=(V,E)$ with grounded Laplacian matrix  $\LL_{-S}$ induced by the grounded node set $S \subset V$,  then for any node $i\in V\backslash S$, its degree $d_i$ is not less than the smallest eigenvalue $\lambda(\LL_{-S})$ of $\LL_{-S}$.
\end{lem}
\begin{proof}
Let $(\lambda(\LL_{-S}),\uu)$ be the smallest eigen-pair of matrix $\LL_{-S}$.  Then, for any unit vector $\vvv$, we have
\begin{equation*}
\vvv^\top (\LL_{-S} -\lambda(\LL_{-S}) \II) \vvv \geq \uu^\top (\LL_{-S} -\lambda(\LL_{-S}) \II) \uu =0.
\end{equation*}
Replacing $\ee_i$ with $\vvv$ leads to
\begin{equation*}
 d_i = (\LL_{-S})_{i,i} \geq \lambda(\LL_{-S})\,,
\end{equation*}
which completes the proof.
\end{proof}

%According to

%For a given node set $S \subset V$, we use $\LL(S)$ to denote the grounded Laplacian matrix induced by $S$, which is obtained by removing rows and columns from the matrix $\LL$ corresponding to the nodes in $S$. And we use the notation $\calG[V\backslash S]$ to denote the subgraph with node set $V\backslash S$ and edge set $(V\backslash S)\times (V\backslash S) \cap E$. For a symmetric semi-positive matrix $\SS$, we use the notation $\lambda(\SS)\geq 0$ to denote the smallest eigenvalue of $\SS$.

\subsection{Concepts related to set function}

We give a brief introduction to some concepts about set functions. %For a set $S$, we use $S+u$ to denote $S\cup\{u\}$, and $S-u$ to denote $S\backslash\{i\}$.
%Then, the definitions of monotonicity and submodularity are given in~\ref{def:mono} and~\ref{def:sub}, respectively.
\begin{definition}[Monotonicity]\label{def:mono}
For set $Q$, a set function $f:2^{Q}\rightarrow \mathbb{R}$ is  monotone non-decreasing if $f(H) \le f(T)$ holds for all $H \subseteq T \subseteq Q$.
\end{definition}

\begin{definition}[Submodularity]\label{def:sub}
A set function $f:2^{Q}\rightarrow \mathbb{R}$ is said to be submodular if
\begin{equation*}\label{submod}
f(H \cup \{u\}) - f(H) \geq f(T\cup \{u\}) - f(T)
\end{equation*}
 holds for all
$H \subseteq T \subseteq Q$ and $u \in Q\backslash T$.
\end{definition}

%\subsection{Submodularity Ratio and Curvature}
For those functions without submodularity, one can define some  quantities to observe the gap between them and submodular functions.

%Besides submodularity, for a set function $f(\cdot)$, its submodularity ratio $\gamma$, and the curvature $\alpha$ are defined as follows.

\begin{definition}[Submodularity ratio~\cite{AbDa11}]
For a non-negative set function $f:2^{Q}\rightarrow \mathbb{R}$, its submodularity ratio is defined as the largest scalar $\gamma$ satisfying
\begin{equation*}
\sum_{u \in T\backslash H} (f(H \cup \{u\})-f(H)) \geq \gamma (f(T)-f(H)), \forall H\subseteq T \subseteq Q.
\end{equation*}
%The submodularity ratio of a non-negative set function $f(\cdot)$ with respect to a parameter $k$ and a set $Q$ is the largest scalar $\gamma$ s.t.
%\begin{equation}
%\gamma= \min_{P \subseteq Q, S: |S|\leq k, S \cap P=\emptyset}\frac{\sum_{x\in S}f(P\cup \{x\})-f(P)}{f(P \cup S)-f(P)}.
%\end{equation}
%$0\leq \gamma \leq 1$, and $f(\cdot)$ is submodular iff $\gamma=1$.
\end{definition}

\begin{definition}[Curvature~\cite{BiBuKrTs17}]
For a non-negative set function $f:2^{Q}\rightarrow \mathbb{R}$, its curvature is defined as the smallest scalar $\alpha$ satisfying
\begin{equation*}
f(T)-f(T\backslash \{u\}) \geq (1-\alpha)(f(H)-f(H\backslash \{u\})), \forall H\subseteq T\subseteq Q, u \in H.
\end{equation*}
%The curvature of a non-negative function $f(\cdot)$ with respect the a set $Q$ is the smallest scalar $\alpha$ s.t.
%\begin{equation}
%\alpha = \max_{S,W\subseteq Q,i \in S \backslash W} 1-\frac{f(S\cup W)-f(S\cup W-i)}{f(S)-f(S-i)}.
%\end{equation}
%$0\leq \alpha \leq 1$, and $f(\cdot)$ is supermodular iff $\alpha=0$.
\end{definition}

%\subsection{Convergence rate*****}

%For a graph $\calG=(V,E)$, its Laplacian matrix $\LL$, a grounded node set $S \subset V$, and its grounded Laplacian matrix $\LL_{-S}$ induced by $S$, the convergence rate of it is the smallest eigenvalue of $\LL_{-S}$, which is denoted by $\lambda(\LL_{-S})$.

\section{Problem formulation}
\label{section:formulation}

The smallest eigenvalue $\lambda(\LL_{-S})$ of  grounded Laplacian matrix $\LL_{-S}$ is a good measure in many  application scenarios. For example, it can be used to characterize the convergence speed of a leader-follower system~\cite{RaJiMeEg09}, the effectiveness of pinning control~\cite{LiXuLuChZe21}, robustness to disturbances of the leader-follower system via the $\mathcal{H}_\infty$ norm~\cite{PiShFiSu18}, to a name a few. In these practical aspects, a larger value of $\lambda(\LL_{-S})$ indicates that the systems have a better performance. As will be shown later, adding edges to graph $\calG$ will lead to an increase of $\lambda(\LL_{-S})$,  which motivates us to propose and study the problem of how to maximize $\lambda(\LL_{-S})$  by creating $k$ new edges.

\subsection{Monotonicity}

For a connected undirected graph $\calG=(V,E)$,  if a set $T$ of new edges is added to $\calG$ from a candidate edge set $Q$ containing all non-existing edges, we obtain a new graph $\calG(T)=(V,E\cup T)$. Let $\LL_{-S}(T)$  be the grounded Laplacian matrix of graph $\calG(T)$,  and let $\lambda(T)=\lambda(\LL_{-S}(T))$ be the smallest eigenvalue of $\LL_{-S}(T)$.

\begin{lem}\label{lem:mon}
For a connected graph $\calG=(V,E)$ with grounded Laplacian matrix $\LL_{-S}$  induced by grounded nodes in set $S$, if we add an edge $e\notin E$ to $\calG$ forming graph $\calG(\{e\})$, then
\begin{equation}
    \lambda( \{e\}) \geq \lambda(\emptyset)=\lambda(\LL_{-S}).
\end{equation}
\end{lem}

\iffalse
\begin{proof}
We prove the lemma by distinguishing three cases: (i) $e$ connects two grounded nodes, (ii) $e$ connects two non-grounded nodes, and (iii) $e$ connects a grounded node and a non-grounded node.

For the first case, $\lambda( \{e\})=\lambda(\LL_{-S})$.

For the second case that a new edge $e_1=(i,j)$ linking two non-grounded nodes $i,j \in V\backslash S$ is added to  $\calG$, let $\PP=-\EE_{i,j}-\EE_{j,i}+\EE_{i,i}+\EE_{j,j}$, and define the smallest eigen-pair of the matrix $\LL_{-S}$ and $\LL_{-S}+\PP$ are $(\lambda,\uu)$ and $(\mu,\vvv)$, respectively. Then,
\begin{equation*}
\mu = \vvv^\top (\LL_{- S}+ \PP) \vvv = \vvv^\top \LL_{-S} \vvv+ (\vvv_i-\vvv_j)^2
   \geq \uu^\top \LL_{-S} \uu =\lambda.
\end{equation*}

%Next, for an edge that connects a grounded node and an ungrounded node, we consider two different cases: can grounded nodes isolate other nodes into several disconnected subgraphs or not.

For the third case that a new edge  $e_2=(t,q)$, $t \in V\backslash S, q\in S$ is added to  $\calG$, the grounded Laplacian matrix of graph $\calG( \{e_2\})$ is $\LL_{- S}+\EE_{t,t}$. Define the smallest eigen-pair of the matrix $\LL_{-S}$ and $\LL_{- S}+\EE_{t,t}$ to be $(\lambda,\uu)$ and $(\eta,\ww)$ respectively. Then
\begin{equation*}
\eta = \ww^\top (\LL_{-S}+\EE_{t,t}) \ww = \ww^\top \LL_{-S} \ww+ \ww_t^2
 \geq \uu^\top \LL_{-S} \uu =\lambda,
\end{equation*}
completing the proof.
\end{proof}
\fi

We omit the proof of Lemma~\ref{lem:mon}, since it is similar to that in~\cite{Fi73}.

Lemma~\ref{lem:mon} indicates that adding an edge $e \in (V\times V)\backslash E$ can increase the smallest eigenvalue of the grounded Laplacian matrix. Below we show that the increase of the smallest eigenvalue is not necessarily strict, by distinguishing three cases of $e$.
%as in the proof of Lemma~\ref{lem:mon}.

%\textcolor{blue}{Although Lemma~\ref{lem:mon} establishes the monotonicity of the smallest eigenvalue of the grounded Laplacian matrix, this monotonicity is not always strict. In some cases, adding an edge does not increase the smallest eigenvalue, which is not the desired outcome. Our goal is to maximize the smallest eigenvalue, so we do not want to consider edges that do not affect the smallest eigenvalue, i.e., edges that do not change the eigenvalue when added. To identify such cases, we analyze the conditions for strict monotonicity by distinguishing three cases of $e\in (V\times V)\backslash E$.}

For the first case, when $e$ connects two grounded nodes,  $\lambda(\{e\})=\lambda(\LL_{-S})$.

For the second case that $e$ connects two non-grounded nodes, consider the graph  $\calG$ in Fig.~\ref{fig:nostr}, where node 0 is the grounded node, and nodes 1 to 4 are non-grounded nodes.  The five solid lines are existing edges in $\calG$, and the dashed line $e=(2,3)$ is a candidate edge, which will be added to  $\calG$. After simple computation, we have $\lambda(\{e\})=\lambda(\LL_{-0})=0.1864$, implying that
adding an edge linking two non-grounded does not necessarily lead to a strict increase of the smallest eigenvalue of the grounded Laplacian matrix.

\begin{figure}[!htbp]
    \centering
    \includegraphics[width=0.5\textwidth]{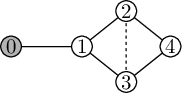}
    \caption{A graph with 5 nodes, with the grey node being the grounded one.}
    \label{fig:nostr}
\end{figure}

%\setlength{\parskip}{0.2cm plus4mm minus5mm}

%\subsection{Add link between followers and a leader}

%Next, for an edge $e=(i,s)$, $i\in V \backslash S, s \in S$, we consider two different cases: can grounded nodes isolate other nodes into several disconnected subgraphs or not.

For the third case that $e$ connects a grounded node and a non-grounded node, we further distinguish two sub-cases:
(i) subgraph $\calG[V\backslash S]$ is connected, and (ii) subgraph $\calG[V\backslash S]$ is disconnected. If  $\calG[V\backslash S]$ is connected, according to the result in~\cite{Fi73}, %from the proof of Lemma~\ref{lem:mon}
we  obtain that the smallest eigenvalue of grounded Laplacian matrix grows strictly.

\iffalse
If the subgraph $\calG[V\backslash S]$ is not connected, then all the elements in $\uu$ are no less than 0. Thus we can obtain the monotonicity of the smallest eigenvalue under this situation, although not strictly, which is
\begin{align*}
\eta & = \ww^\top (\LL_{-S}+\EE_{i,i}) \ww = \ww^\top \LL_{-S} \ww+ \uu_i^2 \\
& \geq \uu^\top \LL_{-S} \uu =\lambda.
\end{align*}
\fi

%To further investigate the monotonicity under the above conditions, w
If  $\calG[V\backslash S]$ is disconnected, suppose that  $\calG[V \backslash S]$ consists of $r>1$ connected subgraphs. Then  $\LL_{-S}$ can be rewritten in  block form as
\begin{equation}
\LL_{-S}=\begin{pmatrix}
\LL_1 & \boldsymbol{0}& \boldsymbol{0}\\
\boldsymbol{0} & \ddots & \boldsymbol{0}\\
 \boldsymbol{0}& \boldsymbol{0}& \LL_r
 \end{pmatrix}.
\end{equation}
It is easy to observe that $\lambda(\LL_{-S})=\min\{\lambda(\LL_1),\dots,\lambda(\LL_r)\}$. According to the above analysis, when we add a new edge $e=(i,t)$ with $i\in V\backslash S$ and $t\in S$, the smallest eigenvalue of the block matrix including node $i$ will strictly increase, while the smallest eigenvalues of other blocks remain unchanged. Thus the addition of edge $e$ will lead to  strict increase of the smallest eigenvalue if and only if the block matrix containing node $i$ has the unique minimum eigenvalue among all blocks.

%From Lemmas~\ref{lem:mon1},~\ref{lem:mon2} and~\ref{lem:mon3}, we can see that as soon as we add new edges to the network, the smallest eigenvalue of the grounded Laplacian matrix will definitely increase. However, only when the new added edge connects a grounded node and a ungrounded node while grounded node set $S$ does not split the rest of the network into several disconnected parts, the smallest eigenvalue will strictly increase.

\subsection{Problem statement}
%Considering only the constraints of edges, we start with the following question of adding $k$ edges to maximize the convergence rate of a social network.
%Consider the case where there are $|S|=s$ leaders in the social network and $|F|=f$ followers, if no follower connects to the leader, its grounded Laplacian matrix degenerates to a Laplacian matrix, whose smallest eigenvalue is 0, and leader's opinion may not spread in the social network, indicating that the convergence rate is 0. And if every followers are connected to every leaders, the convergence rate of this network is $s$. As we will see later, the convergence rate of a social network increases when a new communication link between any two agents is created.

Lemma~\ref{lem:mon} shows that adding an edge will increase  the smallest eigenvalue of the grounded Laplacian matrix. Thus, we put forward the following eigenvalue maximization problem: How to select $k$ edges from a candidate edge set $Q$ to graph $\calG=(V,E)$ with a grounded node set $S$, so that the smallest eigenvalue of the grounded Laplacian matrix generated by set $S$ is maximized. %The specific description of the problem is shown below.
%So we give two problems in this section from two perspectives: maximize the convergence rate when the capacity of the newly added edge set is bounded, and minimize the capacity of the newly added edge set when the convergence rate should be greater than the given threshold.

\begin{tcolorbox}
\begin{problem}\label{prob:1}
Given an undirected unweighted  connected network $\calG=(V,E)$, a grounded node set $S\subset V$ with $|S|\ll n$, a candidate edge set $Q = (V\times V)\backslash E$ consisting of all nonexistent edges, and an integer $k\ll n$, we aim to find an edge set $T\subseteq Q$ with $k$ edges, so that the smallest eigenvalue $\lambda(T)$ for the grounded Laplacian matrix $\LL_{-S}(T)$ of the new graph $\calG(T)=(V,E \cup T)$ is maximized. Mathematically, the problem is formulated as
\begin{equation*}
  	 T^*=\arg\max_{T \subset Q,|T|= k} \lambda(T).
\end{equation*}
\end{problem}
\end{tcolorbox}

%\begin{center}
%\fbox{\parbox{.44\textwidth}{\begin{center}
%\begin{problem}\label{prob:1}
%Given a unweighted, undirected and connected network $\calG=(V,E)$, a grounded node set $S\subset V$, $|S|\ll n$, a candidate edge set $Q = (V\times V)\backslash E$ consisting of all nonexistent edges, and an integer $k\ll n$, we aim to find an edge set $T\subseteq Q$ with $k$ edges, so that for the new graph $\calG(T)=(V,E \cup T)$ and its grounded Laplacian matrix $\LL_{-S}(T)$, its smallest eigenvalue $\lambda(T)$ is maximized. The problem can be formulated as
%\begin{equation*}
%  	 T^*=\arg\max_{T \subset Q,|T|= k} \lambda(T).
%\end{equation*}
%\end{problem}\end{center}
%}}
%\end{center}

When $\calG=(V,E)$ is a small graph, $\lambda(\LL_{-S})$ is large enough for practical applications. In what follows, we focus on large graphs, even those with more than one million nodes. In a large-scale network, if there is an edge between every pair of nodes $u \in S$ and $v \in V\backslash S$, it is easy to verify that $\lambda(\LL_{-S})$ is not less than 1~\cite{GhSr14}, with the lower bound 1 achieved in the star graph whose center is the unique grounded node, which is already sufficiently large for application purposes and thus omitted hereafter. Finally, for a large-scale network where the number of inexistent edges linking nodes $u \in S$ and $v \in V\backslash S$ is smaller than $k$, $\lambda(\LL_{-S})$ is generally large enough for applications~\cite{PiSu16}, which we also do not consider. Therefore, in the sequel, we only treat large graphs, excluding the above two particular cases that are not significant for practical applications.

Problem~\ref{prob:1} is a combinatorial optimization problem. To get its optimal solution, a straightforward approach is to exhaust all $\tbinom{|Q|}{k}$ possible cases for set $T$, calculate the smallest eigenvalue $\lambda(T)$ for each edge set $T$, and then output the optimal solution edge set $T^*$, whose addition to graph $\calG$ leads to the largest smallest eigenvalue $\lambda(T^*)$. The calculation of the smallest eigenvalue takes $O(n^3)$ time for each $T$. For a sparse graph, the size of candidate edge set $Q$ is $O(n^2)$. Thus, the computation complexity is $O\left(\tbinom{|Q|}{k}n^3\right)$, which  is intractable even for small networks.

On the other hand, the objective function $\lambda(\cdot)$ of Problem~\ref{prob:1} is non-submodular as stated in Theorem~\ref{lem:submodular}.

\begin{thm}\label{lem:submodular}
The objective function $\lambda(\cdot):2^{Q}\rightarrow \mathbb{R}$ is a non-submodular function. That is,  there exists two edge sets $A$ and $B$ obeying $A \subseteq B \subseteq Q$ and an edge $e\in Q\backslash B$, which satisfies
\begin{equation}
    \lambda(A\cup \{e\}) -\lambda(A) < \lambda(B\cup \{e\})-\lambda(B).
\end{equation}
\end{thm}

\begin{proof}
To show its non-submodularity, consider an example of a line graph with 8 nodes in Fig.~\ref{fig:nosub}. Let node 0 be the grounded node, and  let $A$ = $\emptyset$, $B$ = $\{(0,2)\}$ and $e$ = $(0,6)$. After simple computation, we have
\begin{align*}
\lambda(A)=0.0437, & &  \lambda(A\cup \{e\})=0.1837, \\
\lambda(B)=0.0644, & &  \lambda(B\cup \{e\})=0.2544.
\end{align*}
Thus,
\begin{equation*}
\lambda(A\cup \{e\})-\lambda(A)=0.14<0.19=\lambda(B\cup \{e\})-\lambda(B).
\end{equation*}
This finishes the proof.
\end{proof}

\begin{figure}[htbp]
    \centering
    \includegraphics[width=0.8\textwidth]{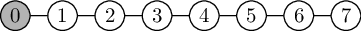}
    \caption{A line graph with 8 nodes and 7 edges. }
    \label{fig:nosub}
\end{figure}
% \begin{proof}
% To show its non-submodularity, consider an example of a line graph with 8 nodes in Fig.~\ref{fig:nosub}. Let node 0 be the grounded node, and  let $A$ = $\emptyset$, $B$ = $\{(0,2)\}$ and $e$ = $(0,6)$. After simple computation, we have
% \begin{align*}
% \lambda(A)=0.0437, & &  \lambda(A\cup \{e\})=0.1837, \\
% \lambda(B)=0.0644, & &  \lambda(B\cup \{e\})=0.2544.
% \end{align*}
% Thus,
% \begin{equation*}
% \lambda(A\cup \{e\})-\lambda(A)=0.14<0.19=\lambda(B\cup \{e\})-\lambda(B).
% \end{equation*}
% This finishes the proof.
% \end{proof}

% \begin{figure}[htbp]
%     \centering
%     \includegraphics[width=0.3\textwidth]{line.eps}
%     \caption{A line graph with 8 nodes and 7 edges. }
%     \label{fig:nosub}
% \end{figure}

%For a combinatorial optimization problem, the submodularity of the objective function can help us design a greedy algorithm by selecting one optimal element in each iteration which has maximum margin benefit, and the solution of this algorithm has a $1-e^{-1}$ approximation ratio~\cite{NeWoFi78}. However, in Problem~\ref{prob:1}, this greedy algorithm has $O(n^5)$ time complexity, and its non-submodularity cannot guarantee a $(1-e^{-1})$-approximation solution. Both of these factors significantly increase the difficulty of the problem.

 %Since $O(\tbinom{n^2}{k})$ cases need to be calculated, this approach is not feasible even on very small networks.

%Such exponential complexity would make exhaust search intractable even in mid-sized networks.

\section{Simplification of the Problem~\ref{prob:1}}

Since for each edge $e \in (S\times S)\backslash E$, $\lambda(\{e\})=\lambda(\LL_{-S})$, indicating that adding $e$ has no influence  on the smallest eigenvalue of the grounded Laplacian matrix. Thus, candidate set $Q=(V\times V)\backslash E$ in Problem~\ref{prob:1} reduces to $(V \times (V \backslash S))\backslash E$, which, however, still has $O(n^2)$ edges. This is one of the challenges for solving Problem~\ref{prob:1}. Below we will show that Problem~\ref{prob:1} can be further simplified by restricting $Q$ to a subset with $O(n)$ edges, without alerting the optimal solution. The core idea is to prune a large portion of insignificant edges in $Q$.

\begin{lem}\label{lem:rep}
Given a graph $\calG=(V,E)$, a grounded node set $S \subset V$,  grounded Laplacian matrix $\LL_{-S}$, a candidate edge set $Q= (V \times V)\backslash E$, for an edge $e=(i,j)\in Q$ with $i,j \in V\backslash S$, there must exist an edge $y=(z,x)\in Q$ with $z\in V\backslash S$ and $x \in S$, satisfying
\begin{equation}
    \lambda(\{y\}) \geq \lambda(\{e\}).
\end{equation}
%three nodes $i,j\in V\backslash S$, $s\in S$, and three edges $e_1=(i,j)$, $e_2=(i,s)$, $e_3=(j,s)$, the smallest eigenvalue of the grounded Laplacian matrix after adding $e_2$ or $e_3$ to the graph $\calG$ will not both be smaller than that after adding $e_1$ to the graph $\calG$ under a certain approximation.
\end{lem}

\begin{proof}
Let  $(\lambda,\uu)$ be the smallest eigen-pair of grounded Laplacian matrix $\LL_{-S}$ for graph $\calG=(V,E)$. When we add a new edge $e=(i,j)$, $i,j \in V\backslash S$, to $\calG$,  the grounded Laplacian matrix $\LL_{-S}$ is perturbed with  matrix $\PP= \EE_{i,i}+\EE_{j,j}-\EE_{i,j}-\EE_{j,i}$, and the smallest eigenvalue and eigenvector $\lambda$ and $\uu$ varies, separately, with $\mu$ and $\vvv$. Thus, we obtain
\begin{equation}
(\LL_{-S}+\PP)(\uu + \vvv)=(\lambda+ \mu)(\uu +\vvv).
\end{equation}
Left multiplying $\uu^\top$ on both sides of the above equation gives
\begin{equation}
\mu=\frac{\uu^\top \PP (\uu+\vvv)}{\uu^\top (\uu+\vvv)}.
\end{equation}

According to matrix perturbation theory~\cite{Mi11,Na06}, the small perturbation $\PP$ on the matrix $\LL_{-S}$ has little influence on the eigenvector $\uu$, especially the matrix $\LL_{-S}$ is large scaled matrix, which means $\vvv \approx \boldsymbol{0}$~\cite{HeYaYuZh19,LiCaHeChLiChSu15,MiSuNi10,ReOtHu06}. Hence,
\begin{equation}\label{equ:inc}
 \mu \approx (\uu_i-\uu_j)^2.
\end{equation}

In order to prove the lemma, we distinguish the following three cases.

%In order to find an edge that connects a grounded node and an ungrounded node, whose addition leads to a larger increment of the smallest eigenvalue than in Eq.~(\ref{equ:inc}), we will discuss this in three cases.

 %the increase of the smallest eigenvalue can be approximated by $\delta_e \approx (\uu_i-\uu_j)^2$.

For the first case that there exist edges $e_1=(i, s_1)\in Q$, $e_2=(j, s_2)\in Q$, $s_1, s_2 \in S$,  adding $e_1$ or $e_2$ to  graph $\calG$ will lead to the grounded Laplacian matrix $\LL_{-S}$ perturbed by $\EE_{i,i}$ or $\EE_{j,j}$. Let $\xi$ and $\phi$ denote, respectively, the increment of the smallest eigenvalue $\lambda$ after adding $e_1$ and $e_2$. Using a similar analysis as above, we  obtain
\begin{equation}\label{equ:inc2}
\xi \approx \uu_i^2, \phi \approx \uu_j^2.
\end{equation}
Since $\uu_i\geq 0$ and $\uu_j \geq 0$, $\max\{\uu_i^2,\uu_j^2\}\geq (\uu_i-\uu_j)^2$ holds for any  $e=(i,j)\in Q$ with  $i,j \in V \backslash S$, which indicates that the maximal eigengain induced by  edge $e_1$ or $e_2$ is not less than that of $e$.

\iffalse
If there exist grounded nodes $s_1, s_2 \in S$, and edges $e_1=(i, s_1)\in Q$, $e_2=(j, s_2)\in Q$, then adding edges $e_1$ or $e_2$ to the graph $\calG$, the grounded Laplacian matrix $\LL_{-S}$ is perturbed with a matrix $\EE_{i,i}$ or $\EE_{j,j}$. The increment of the eigenvalues is denoted by $\xi$ and $\phi$, respectively. Using a similar analysis, we can obtain
\begin{equation}
\xi \approx \uu_i^2, \phi \approx \uu_j^2.
\end{equation}
\fi

For the second case only that there still exists an edge $e_2=(j,s_2)\in Q$ (resp.  $e_2=(i,s_2)\in Q$), $s_2\in S$, while node $i$  (resp. $j$) connects  all grounded nodes, according to~\cite{Pi14}, there exists one node $p\in V\backslash S$, and an edge $e_3=(p,s_3)\in Q$, $s_3\in S$, such that $\uu_p \geq \uu_i$ (resp. $\uu_p \geq \uu_j$). It is easy to derive that the increment of the smallest eigenvalue induced $e_3$ can be approximated by $\uu_p^2$.  Hence, $\max\{\uu_p^2,\uu_j^2\}\geq \max\{\uu_i^2,\uu_j^2\}\geq (\uu_i-\uu_j)^2$, indicating that the maximal eigengain of edges $e_2$ and $e_3$ is not smaller than that of $e$.

For the third case that both nodes $i$ and $j$ are connected to all grounded nodes, according to the result in~\cite{Pi14}, there exists one node $p\in V\backslash S$, and an edge $e_4=(q,s_4)\in Q$, $s_4\in S$, such that $\uu_q \geq \max\{\uu_i,\uu_j\}$. Hence, $\uu_q^2\geq (\uu_i-\uu_j)^2$, indicating edge $e_4$ more important than $e$, with respect to increasing the smallest eigenvalue of grounded Laplacian matrix.

The analysis of the above three cases together complete our proof.
\end{proof}

By Lemma~\ref{lem:rep}, there are a large portion of insignificant edges in the candidate edge set $Q$, each of which links two non-grounded nodes. Thus, Problem~\ref{prob:1} can be simplified by pruning insignificant edges, only keeping those edges in $Q$, each of which is connected to a grounded node and a non-grounded node. Also, the non-grounded node of the candidate edge decides the changes in the eigenvalue. This reduces the size of the candidate edge set $Q$ from $O(n^2)$ to $O(n)$. Then we naturally propose a simplified version of Problem~\ref{prob:1}.

\begin{tcolorbox}
\begin{problem}\label{prob:2}
Under the same assumption as Problem~\ref{prob:1} with the exception of $Q = (S\times (V\backslash S))\backslash E $, we aim to find an edge set $T\subseteq Q$ with $k$ edges, so that the smallest eigenvalue $\lambda(T)$ of the grounded Laplacian matrix $\LL_{-S}(T)$ for the new graph $\calG(T)=(V,E \cup T)$ is maximized. The problem is mathematically formulated as
\begin{equation*}
  	 T^*=\arg\max_{T \subset Q,|T|= k} \lambda(T).
\end{equation*}
\end{problem}
\end{tcolorbox}

The only difference between Problems~\ref{prob:1} and~\ref{prob:2} is the candidate edge set $Q$. Theorem~\ref{them:same} provides the connection for optimal solutions between these two Problems.

\begin{thm}\label{them:same}
Any optimal solution to Problem~\ref{prob:2} is also an optimal solution for Problem~\ref{prob:1}.
\end{thm}
\begin{proof}
For graph $\calG=(V,E)$ with grounded node set $S\subset V$, let $Q_1=(V\times V)\backslash E$ and $Q_2=(S\times (V\backslash S))\backslash E$, and let $T_1$, and $T_2$ be the optimal solution of  Problem~\ref{prob:1} and Problem~\ref{prob:2}, respectively. Obviously, $\lambda(T_2) \leq \lambda(T_1)$ since $Q_2 \subseteq Q_1$. We next prove $\lambda(T_1) \leq \lambda(T_2)$ by distinguishing two cases: (i) $T_1 \subset Q_2$ and  (ii) $T_1 \subsetneq Q_2$.

For case (i) $T_1 \subset Q_2$, $T_1$ is obviously an optimal solution to Problem~\ref{prob:2}. Thus, $\lambda(T_2)\leq \lambda(T_1) \leq \lambda(T_2)$.

For case (ii) $T_1 \subsetneq Q_2$, there must exist at least one edge $e_1=(i,j)\in T_1$, $i,j \in V\backslash S$.  By Lemma~\ref{lem:rep}, for graph $\calG(T_1\backslash \{e_1\})=(V, E\cup T_1\backslash \{e_1\})$, there exists an edge $r_1=(p,t)\in Q_1 \backslash T_1$, $p\in V\backslash S$, $t \in S$, satisfying $\lambda(T_1\cup\{r_1\}\backslash \{e_1\})\geq \lambda(T_1)$. Define edge set $P_1=T_1\cup\{r_1\}\backslash \{e_1\}$, which is obtained from $T_1$ by replacing edge $e_1 \in T_1$ with edge $r_1$. Then, we have $\lambda(P_1)\geq \lambda(T_1)$. Repeating this replacement at most $k$ times, we get an edge set $P_k\subseteq Q_2$, where no edge is in set $Q_1\backslash Q_2$, satisfying $\lambda(P_k)\geq\lambda(P_1)\geq \lambda(T_1)$. Therefore, we have
\begin{equation}
    \lambda(T_2)\leq \lambda(T_1)\leq \lambda(P_1)\leq \lambda(P_k) \leq \lambda(T_2),
\end{equation}
which finishes the proof.
\end{proof}

To deepen the understanding of Theorem~\ref{them:same}, we present a simple example. Consider the graph in Fig.~\ref{fig:nostr}, where node 0 be the grounded node and the solid lines represent the edges in the original graph. Then, for Problem~\ref{prob:1}, the candidate edge set is $Q_1 = \{(0,2),(0,3),(0,4),(2,3),(1,4)\}$; while for Problem~\ref{prob:2}, the candidate edge set is $Q_2 = \{(0,2),(0,3),(0,4)\}$. We consider the case of $k=2$ for both problems. We calculate the smallest eigenvalue of the grounded Laplacian matrix after adding $2$ edges from each candidate edge set. Below we list the corresponding smallest eigenvalues for all possible sets of edges added.
\begin{align*}
&\lambda(\{(0,2),(0,4)\}) = \lambda(\{(0,3),(0,4)\}) = \lambda(\{(0,2),(0,3)\}) = 1.120,\\
&\lambda(\{(0,2),(2,3)\}) = \lambda(\{(0,3),(2,3)\}) = 1.000, \lambda(\{(0,4),(2,3)\}) = 0.829, \\
&\lambda(\{(0,2),(1,4)\}) = \lambda(\{(0,3),(1,4)\}) = 1.000, \lambda(\{(0,4),(1,4)\}) = 0.824, \\
&\lambda(\{(2,3),(1,4)\}) = 0.697.
\end{align*}
Thus, the optimal solutions for Problem~\ref{prob:2} are the edge sets $\{(0,2),(0,4)\}$, $\{(0,2),(0,3)\}$, or $\{(0,3),(0,4)\}$, which are also the optimal solutions for Problem~\ref{prob:1}.

Now, we can obtain the optimal solution Problem~\ref{prob:1} by exhausting $\tbinom{|Q|}{k}$ possible cases for edge set $T$ in time $O(\tbinom{|Q|}{k}n^3)$. However, this na\"{\i}ve method only works for very small $n$ and $k$.

\section{Simple greedy algorithm}

%Despite the reduction of the candidate edge set $Q$ greatly compresses its size, the na\"{\i}ve method is still intractable in mid-sized network. Also, the example in Fig.~\ref{fig:nosub} shows the objective function is non-submodular even on a smaller candidate edge set. However, the standard greedy algorithm usually has a good performance~\cite{BiBuKrTs17,AbDa11}, and a tight approximation guarantee of $(1-e^{-\gamma\alpha})/\alpha$ is given in~\cite{BiBuKrTs17} where $\gamma$ and $\alpha$ are submodularity ratio and curvature respectively.

%\subsection{Simple Algorithm}

As shown above, for Problem~\ref{prob:1} exhausting possible cases needs an exponential time. In this section, we resort to a simple greedy algorithm to approximately solve Problem~\ref{prob:1}, which is outlined in Algorithm~\ref{alg:1}. Initially, the augmented edge set $T$ is set to be empty. Then $k$ edges are iteratively selected to the augmented edge set from set $ Q\setminus T$. In each iteration, the edge $e$ in candidate set with maximum $\lambda(T\cup \{e\})-\lambda(T)$ denoted by $\lambda_T(e)$ is chosen. The standard greedy algorithm terminates when $k$ edges are chosen to be added to $T$.

%For consistency of structure, we first introduce the greedy algorithm. For a given graph $\calG=(V,E)$, a grounded node set $S \subset V$, a candidate edge set $Q= (S \times (V\backslash S))\backslash E$, and edge set $T\subset Q$, edge $e\in Q\backslash T$, $\lambda(T\cup \{e\})\geq \lambda(T)$ holds true according to its monotonicity. We define $\lambda_T(e)=\lambda(T\cup \{e\})-\lambda(T)$ as the increment of the smallest eigenvalue after adding edge $e$ to the edge set $T$, and $\lambda_T(e)\geq 0$ for any edge $e\in Q\backslash T$. Then the greedy algorithm is outlined in Algorithm~\ref{alg:1}. First, we set the edge set $T$ empty, then we iteratively add an edge from the candidate edge set $Q$ which maximizes $\lambda_T(e)$ to the edge set $T$. The greedy algorithm stops when $k$ edges are added to the edge set $T$.

\begin{algorithm}
	\caption{\textsc{Greedy}$(\calG, k, S)$}
		\label{alg:1}
		\Input{
			A graph $\calG=(V,E)$; an integer $k$; a grounded node set $S$
		}
		\Output{
			$T$: an edge set with $|T| = k$
		}
		Initialize solution $T = \emptyset$ \;
        Initialize candidate edge set $Q= (S \times (V\backslash S))\backslash E$\;
		\For{$i = 1$ to $k$}{
           Compute $\lambda_T(e) = \lambda(T \cup\{e\})-\lambda(T)$ for each $e \in Q$\;
			Select $e$ s.t.  $e \gets \mathrm{arg\, max}_{e \in Q} \lambda_T(e)$ \;
            Update $Q \gets Q \backslash \{e\}$\;
			Update solution $T \gets T\cup \{e\}$ \;
		}
		\Return $T$
\end{algorithm}

%To implement this algorithm, we take $k$ steps, and in each step we calculate the change in the smallest eigenvalue after adding an edge to find the edge that leads to the greatest change of the eigenvalue. The time complexity of calculating the eigenvalue of each edge in line 4 takes $O(nm)$ time and the overall time complexity of the algorithm is $O(knm)$.

% Also, the example in Fig.~\ref{fig:nosub} shows the objective function is non-submodular even on a smaller candidate edge set. However, the standard greedy algorithm usually has a good performance~\cite{BiBuKrTs17,AbDa11}, and a tight approximation guarantee of $(1-e^{-\gamma\alpha})/\alpha$ is given in~\cite{BiBuKrTs17} where $\gamma$ and $\alpha$ are submodularity ratio and curvature respectively.

%The effectiveness and the time complexity of Algorithm~\ref{alg:1} are stated in the following Theorem~\ref{them:appx}.

Although the objective function of Problem~\ref{prob:2} is not submodular, the standard greedy algorithm usually has a good performance~\cite{BiBuKrTs17,AbDa11} with a tight approximation guarantee of $(1-e^{-\gamma\alpha})/\alpha$~\cite{BiBuKrTs17} where $\gamma$ and $\alpha$ are submodularity ratio and curvature, respectively. The effectiveness and the time complexity of Algorithm~\ref{alg:1} are summarized in  Theorem~\ref{them:appx}.

\begin{thm}\label{them:appx}
Let $\gamma \in [0,1]$ and $\alpha \in [0,1]$ be submodularity ratio and curvature of the objective function of  Problem~\ref{prob:2}. Then, Algorithm~\ref{alg:1} runs in time $O(kn^4)$, and edge set $T$ returned by Algorithm~\ref{alg:1} satisfies
%the answer $T=\textsc{Greedy}(\calG,k,S)$ returned by Algorithm~\ref{alg:1} which runs in time $O(kn^4)$ has the following approximation guarantee~\cite{BiBuKrTs17}:
\begin{equation}
\lambda(T) \geq \frac{1}{\alpha}(1-e^{-\alpha\gamma})\lambda(T^*),
\end{equation}
where %$T^*$ is the optimal solution of the problem,
\begin{equation*}
T^* =\arg\max_{T \subseteq Q, |T|=k} \lambda(T).
\end{equation*}
\end{thm}
\begin{proof}
%The time complexity of Algorithm~\ref{alg:1} is easy to check.
In line 4 of Algorithm~\ref{alg:1}, for each candidate edge $e\in Q$,  calculating the  smallest eigenvalue $\lambda(T \cup \{e\})$ needs $O(n^3)$ time. Since there are $O(n)$ edges in the candidate edge set,  the running time of computing  $\lambda_T(e)$ for all $e\in Q$ is $O(n^4)$. Picking out the edge $e$ with the largest $\lambda_T(e)$ takes $O(n)$ time. Thus, the overall time complexity of the Algorithm~\ref{alg:1} to pick $k$ edges is $O(kn^4)$.

The proof of the ratio $(1-e^{-\gamma\alpha})/\alpha$ for approximation guarantee is omitted, since it is similar to that in~\cite{BiBuKrTs17}.
\end{proof}

%\subsection{Bounds of Submodularity Ratio and Curvature}
%Although the objective function is not submodular, inspired by~\cite{AbDa11,BiBuKrTs17}, we can get the submodularity ratio, and the curvature of the function, to make the subsequent heuristic algorithms have some guarantees.

%Only when the submodularity ratio and curvature of the objective function can be bounded, the greedy Algorithm~\ref{alg:1} and its approximation ratio in Theorem~\ref{them:appx} are meaningful rather than a mere heuristic. However, it is not easy to bound these two quantities. So next in Lemma~\ref{lem:bound}, we give bounds of these two quantities when the subgraph $\calG[V \backslash S]$ is connected, while the bounds in general cases are open to further discussion.

%We will show that the these quantities of the objective function can be bounded when $G[V \backslash S]$ is a connected graph, so the approximation guarantees from theoretical results are applicable. Thus the greedy strategies can be a principled choice rather than a mere heuristic, and an approximation guarantee of the greedy strategy has been given in.

%In next part, we will give the bounds of submodularity ratio, and the curvature of the function, for the convenience of the calculation, we consider the case where $\calG[V\backslash S]$ is a connected graph.

It is difficult to determine the submodularity ratio $\gamma$ and the curvature $\alpha$ for the objective function of Problem~\ref{prob:2} for general cases. However, when graph $\calG[V \backslash S]$ is connected, we can provide a lower bound and an upper bound for $\gamma$ and $\alpha$, respectively.

%\subsubsection{Bounding the submodularity ratio}
\begin{thm}\label{lem:bound}
For a graph $\calG=(V,E)$ with grounded node set $S\subset V$, candidate edge set $Q = (S\times (V\backslash S))\backslash E$, let $\bar{\LL}$ be the Laplacian matrix of its subgraph $\calG[V\backslash S]$, and let $\lambda_2(\bar{\LL})$ be smallest non-zero eigenvalue of $\bar{\LL}$. Then, when $\calG[V\backslash S]$ is connected, for the function $\lambda(\cdot):2^{Q}\rightarrow \mathbb{R}$ of Problem~\ref{prob:2}, the submodularity ratio $\gamma$ is bounded by:
\begin{equation}
\gamma \geq \frac{1}{sn} \left(1-\frac{2s\sqrt{(n-s)}}{\lambda_2(\bar{\LL})}\right)^2,
\end{equation}
and the curvature $\alpha$ is bounded by:
\begin{equation}
\alpha \leq 1-\frac{1}{sn} \left(1-\frac{2s\sqrt{(n-s)}}{\lambda_2(\bar{\LL})}\right)^2.
\end{equation}
\end{thm}
\begin{proof}
We first provide a lower bound for $\gamma$. To this end, for
\begin{equation} \label{GammaLow}
\frac{\sum_{x \in T\backslash H}(\lambda(H\cup \{x\})-\lambda(H))}{\lambda(T)-\lambda(H)}, \forall H\subset T \subseteq Q.
\end{equation}
we provide a lower bound and an upper for its numerator and denominator, respectively.

Let $\uu$ be the eigenvector corresponding to $\lambda(H)$ of matrix $\LL_{- S}(H)$, and let $\uu^{(x)}$ be the eigenvector corresponding to $\lambda(H \cup \{x\})$ of matrix $\LL_{- S}(H\cup \{x\})$. Moreover, let $x_i$ be the non-grounded node incident to edge $x$. According to matrix perturbation theory~\cite{HeYaYuZh19,MiSuNi10,ReOtHu06}, we have
\begin{equation}
\lambda(H\cup \{x\}) - \lambda(H) = \frac{\uu_{x_i} \uu_{x_i}^{(x)}}{\uu^\top \uu^{(x)}}.
\end{equation}

For the numerator of~\eqref{GammaLow}, we have
\begin{align}\label{equ:gamnum}
& \sum_{x \in T\backslash H}(\lambda(H\cup \{x\})-\lambda(H)) =  \sum_{x \in T\backslash H} \frac{\uu_{x_i}\uu^{(x)}_{x_i}}{\uu^\top \uu^{(x)}} \nonumber\\
& \geq  \sum_{x \in T\backslash H} \frac{\uu_{x_i}\uu^{(x)}_{x_i}}{\norm{\uu}_2\smallnorm{\uu^{(x)}}_2}  \geq  \sum_{x \in T\backslash H} (\uu^{(x)}_{\min})^2 \nonumber\\
& = \sum_{x \in T\backslash H} \frac{(\uu^{(x)}_{\min})^2}{(\uu^{(x)}_{\max})^2} (\uu^{(x)}_{\max})^2
\geq \frac{1}{n} \left(1-\frac{2s\sqrt{(n-s)}}{\lambda_2(\bar{\LL})}\right)^2,
\end{align}
where $\uu^{(x)}_{\max}$ and $\uu^{(x)}_{\min}$ are, respectively, the largest and the smallest components of vector $\uu^{(x)}$, and the last inequality is established  based on the results in~\cite{PiShSu15,Pi14}.

For the denominator of~\eqref{GammaLow}, we have
\begin{equation}\label{equ:gamden}
\lambda(T)-\lambda(H) \leq  \lambda(Q) -\lambda(\emptyset) \leq s.
\end{equation}

Combining~\eqref{equ:gamnum} and~\eqref{equ:gamden} leads to
\begin{equation}
\gamma \geq \frac{1}{sn} \left(1-\frac{2s\sqrt{(n-s)}}{\lambda_2(\bar{\LL})}\right)^2\,.
\end{equation}
%which completes the first part.

In order to give an upper bound for $\alpha$, we provide, respectively, a lower bound and an upper for the numerator and denominator of
$\frac{\lambda(T)-\lambda(T  \backslash \{u\})}{\lambda(H)-\lambda(H \backslash \{u\})}$, $\forall H\subseteq T \subseteq Q, u \in H$.
Utilizing an approach similar to~\eqref{equ:gamnum} and~\eqref{equ:gamden}, we obtain
\begin{equation}\label{equ:alpnum}
\lambda(T)-f(T \backslash \{u\}) \geq  \frac{1}{n} \left(1-\frac{2s\sqrt{(n-s)}}{\lambda_2(\bar{\LL})}\right)^2
\end{equation}
and
\begin{equation}\label{equ:alpden}
\lambda(H)-\lambda(H \backslash \{u\}) \leq \lambda(Q) -\lambda(\emptyset) \leq s,
\end{equation}
both of which lead to
\begin{equation}
\alpha \leq 1-\frac{1}{sn} \left(1-\frac{2s\sqrt{(n-s)}}{\lambda_2(\bar{\LL})}\right)^2.
\end{equation}
This completes the proof.
\end{proof}

Note that the approximation ratio $(1-e^{-\gamma\alpha})/\alpha$ characterizes the theoretical performance of the greedy algorithm, with larger $(1-e^{-\gamma\alpha})/\alpha$ corresponding to better performance. On the other hand, the approximation ratio $(1-e^{-\gamma\alpha})/\alpha$ depends on $\alpha$ and $\gamma$. It is desirable that $\alpha$ is small and $\gamma$ is large. Theorem~\ref{lem:bound} indicates that the bounds of $\gamma$ and $\alpha$ are related to $\lambda_2(\bar{\LL})$. Concretely, a larger $\lambda_2(\bar{\LL})$ implies a larger lower bound of $\gamma$ and a smaller upper bound of $\alpha$. Since denser graphs usually have a larger $\lambda_2(\bar{\LL})$,  a tighter bound for approximation ratio can be achieved on denser graphs.

\section{Fast greedy algorithm}

%To solve the problem, we put forward three methods from different insights.

The complexity of the simple greedy algorithm~\ref{alg:1} has been significantly reduced compared with the exhaustive search method, but it is still not applicable to large networks since it takes too much time to calculate $\lambda_T(e)$ for each $e\in Q$ in each iteration. On the other hand, although the approximation for $\lambda_T(e)$ in the proof of Lemma~\ref{lem:rep} is enough for proving the Lemma, it is not suitable for estimating the importance of two edges linking non-grounded nodes and grounded nodes, since the difference for their eigengains  may be small that cannot be discriminated by this rough approximation. To overcome this issue, we use an alternative way to evaluate the upper bound for the eigengain of an edge, based on which we further give a fast approximation algorithm to solve Problem~\ref{prob:2} in time $\tilde{O}(km)$. Note that the idea of this bound estimation method has been previously used in the literature~\cite{MoBaZhPe20,ToChToEl21}.

%Although the complexity of the simple greedy algorithm~\ref{alg:1} has been significantly reduced compared with the na\"{\i}ve method, however, the solution is still not fast for large-scale networks since it takes too much time to calculate $\lambda_T(e)$ for each $e\in Q$ in each iteration. And the approximation mentioned in Lemma~\ref{lem:rep} is a rough estimate since getting the increment accurately requires the knowledge of eigenvector of $\LL_{-S}(T\cup\{e\})$. So in this section, we will approximate a upper bound of $\lambda_T(e)$, and the upper bound approximation has been used in another eigengap approximation~\cite{ToChToEl21}. Then we give a fast approximation algorithm to Problem~\ref{prob:2} with time complexity $\tilde{O}(km)$.

\iffalse
In order to get the smallest eigenvalue and eigenvector of the grounded Laplacian matrix with a low computational cost, we first introduce the method mentioned in~\cite{BaSpSrTe13,SpTe14,CoKyMiPaJaPeRaXu14}, and the performance of this method is given in the following lemma.

\begin{lem}
Given a grounded Laplacian matrix $\LL_{-S}$ with $m$ nonzero elements and smallest eigenvalue $\lambda_{\min}(\LL_{-S})$, there exists a nearly-linear time algorithm $\uu=\textsc{Vector}(\LL_{-S},\epsilon)$ that outputs a vector $\uu$ such that $\lambda=\uu^\top \LL_{-S} \uu \leq (1+\epsilon)\lambda_{\min}(\LL_{-S})$. The algorithm takes $\tilde{O}(m)$ time, where $\tilde{O}(\cdot)$ hides the ${\rm poly} (\log n)$ factor.
\end{lem}
\fi

We next provide an upper bound of $\lambda_T(e)$ for an edge in the candidate edge $Q$ of Problem~\ref{prob:2}. For an edge $e=(t,i)\in Q$ that connects a grounded node $t\in S$ and a non-grounded node $i \in V\backslash S$, according to Cauchy’s interlacing theorem~\cite{PaBa10}, we have
\begin{equation*}
\lambda(\LL_{-S}(T)) \leq \lambda(\LL_{-S}(T)+\EE_{i,i}) \leq \lambda((\LL_{-S}(T))_{-i}).
\end{equation*}
Define $\bar{\lambda}_T(e)=\lambda((\LL_{-S}(T)_{-i})-\lambda(\LL_{-S}(T))\geq \lambda_T(e)$, then $\bar{\lambda}_T(e)$ is an upper bound of $\lambda_T(e)$.

To bridge the gap between matrix  $\LL_{-S}(T)$ and matrix $\LL_{-S}(T)_{-i}$, we introduce the following matrix $\MM_i$,
\begin{equation*}
 \MM_i=\begin{pmatrix}
(\LL_{-S}(T))_{-i} & \boldsymbol{0} \\
\boldsymbol{0} & (\LL_{-S}(T))_{i,i}
\end{pmatrix}.
\end{equation*}
If there exists one node $j\in V\backslash S$ whose degree is not larger than the degree of node $i$, then by Lemma~\ref{lem:lmd} the smallest eigenvalue $\lambda(\MM_i)$ of matrix $\MM_i$ is equal to $\lambda((\LL_{\calG}(T))_{-i})$.
Thus, as long as  $i$ is not the unique node in  $V\backslash S$ with the smallest degree,  $\lambda(\MM_i)=\lambda((\LL_{\calG}(T))_{-i})$ always holds. In other words, there exists at most one such node $i$, whose degree is smaller than all other nodes. This situation is almost negligible, since adding a new edge from a grounded node to the node with smallest degree is often not a good choice. Thus, with a probability of at least $1-|S|/|Q|$, $\bar{\lambda}_T(e)=\lambda(\MM_i)-\lambda(\LL_{-S}(T))$ holds. Below, we use $\lambda(\MM_i)-\lambda(\LL_{-S}(T))$ instead of $\bar{\lambda}_T(e)$,  as an upper bound of $\lambda_T(e)$.

The following Lemma provides an approximation of $\lambda(\MM_i)-\lambda(\LL_{-S}(T))$ by using a partial derivative method~\cite{YiSh2018,SiBoBaMo18,KaTo19}.
\begin{lem}\label{lem:fin}
Given a connected graph $\calG=(V,E)$, a grounded node set $S\subset V$, a  candidate edge set $Q = (S\times (V\backslash S))\backslash E$, and $T \subseteq Q$, let $\LL_{-S}(T)$ be the grounded Laplacian matrix of graph $\calG(T)=(V,E \cup T)$, and let  $(\lambda,\uu)$ be the smallest eigen-pair of matrix $\LL_{-S}(T)$. Then, for any node $i\in V\backslash S$, we have the following approximation:
%For a connected graph $\calG=(V,E)$, a grounded node set $S$, a candidate edge set $Q = (S\times (V\backslash S))\backslash E$, $T \subseteq Q$, the graph $\calG(T)=(V,E \cup T)$ is obtained by adding edges from edge set $T$ to the graph $\calG$. Its grounded Laplacian matrix is $\LL_{-S}(T)$, and the smallest eigen-pair of it is defined as $(\lambda,\uu)$. For any node $i\in V\backslash S$, we have the following approximation:
\begin{equation}\label{eigengain}
\lambda(\MM_i)-\lambda(\LL_{-S}(T)) \approx 2\uu_i\sum_{j\in \calN_i\backslash S}\uu_j =\tilde{\lambda}_T(e).
\end{equation}
\end{lem}
\begin{proof}
Notice that matrix $\MM_i$ is obtained from matrix $\LL_{-S}(T)$ by deleting the non-diagonal entries in the row and column corresponding to node $i$. To evaluate $\lambda(\MM_i)-\lambda(\LL_{-S}(T))$, we use the partial derivative~\cite{YiSh2018,SiBoBaMo18,KaTo19} to measure the change of the smallest eigenvalue $\lambda(\LL_{-S}(T))$, caused by a nonzero non-diagonal entry at row $i$ or column $i$.

By definition of eigenvalue, we have $\lambda \LL_{-S}(T)=\lambda \uu$. For any $j\neq i$ and $\LL_{-S}(T)_{i,j}\neq 0$, performing the derivative of both sides of  $\lambda \LL_{-S}(T)=\lambda \uu$ with respect to $\LL_{-S}(T)_{i,j}$ yields
\begin{equation*}
    \frac{\partial \LL_{-S}(T)}{\partial \LL_{-S}(T)_{i,j}} \uu+ \LL_{-S}(T)\frac{\partial \uu}{\partial \LL_{-S}(T)_{i,j}}=\frac{\partial \lambda}{\partial \LL_{-S}(T)_{i,j}} \uu+ \lambda\frac{\partial \uu}{\partial \LL_{-S}(T)_{i,j}}.
\end{equation*}
Left multiplying $\uu^\top$ on both sides, we  obtain
\begin{equation*}
    \frac{\partial \lambda}{\partial \LL_{-S}(T)_{i,j}}=\uu^\top \frac{\partial \LL_{-S}(T)}{\partial \LL_{-S}(T)_{i,j}}\uu = \uu_i\uu_j.
\end{equation*}
%Similarly, we have
%\begin{equation*}
%    \frac{\partial \lambda}{\partial \LL_{-S}(T)_{j,i}}=\uu^\top \frac{\partial \LL_{-S}(T)}{\partial \LL_{-S}(T)_{j,i}}\uu = \uu_j\uu_i.
%\end{equation*}
Adding up all the partial derivative on non-diagonal nonzero element at $i$-th row and $i$-th column of matrix $\LL_{-S}(T)$, we get an approximation of
\begin{equation*}
\lambda(\MM_i)-\lambda(\LL_{-S}(T)) \approx     2\uu_i\sum_{j\in \calN_i\backslash S}\uu_j =\tilde{\lambda}_T(e),
\end{equation*}
which is exactly~\eqref{eigengain}.
%to denote the derivative of the smallest eigenvalue from matrix $\LL_{-S}(T)$ to matrix $\MM_i$. Thus the eigengap between $\lambda(\MM_i)$ and $\lambda(\LL_{-S}(T))$ can be approximated by $\tilde{\lambda}_T(e)$.
\end{proof}

%the strong correlation between $\lambda_T(e)$ and $\tilde{\lambda}_T(e)$ later in the Experiment Section.
%
% We define $\Delta \MM = \sum_{j\in \calN_i \backslash S} \EE_{i,j}+\EE_{j,i}$, and $\MM_i=\LL_{-S}(T)+\Delta \MM$. We let $(\lambda(\LL_{-S}(T)),\uu)$ be the eigen-pair of $\LL_{-S}(T)$ and it varies with $\Delta \lambda$ and $\Delta \uu$ after $\LL_{-S}(T)$ is perturbed with $\Delta \MM$, then we have
%\begin{equation}
%\Delta \lambda = \frac{\uu^\top \Delta \MM (\uu+\Delta \uu)}{\uu^\top (\uu+\Delta \uu)} \approx 2 \uu_i\sum_{j\in \calN_i \backslash S} \uu_j=\tilde{\lambda}_T(i).
%\end{equation}
%
%Thus for each new added edge connects to a follower $i\in F$, we calculate $\tilde{\lambda}_T(i)\approx \bar{\lambda}_T(i)$ to replace the original approximation.

%\subsection{Fast Algorithm}

Lemma~\ref{lem:fin} and the above analysis show that finding the maximum $\lambda_T(e)$ is reduced to maximizing its upper bound $\bar{\lambda}_T(e)$, which can be approximated by $\tilde{\lambda}_T(e)$ for each non-existing edge. Then, to compute $\tilde{\lambda}_T(e)$, the only thing left is to evaluate the smallest eigen-pair of the grounded Laplacian matrix $\LL_{-S}$ with a low computational cost. Fortunately, this can be solved by the method  in~\cite{BaSpSrTe13,SpTe14,CoKyMiPaJaPeRaXu14}, as given in the following lemma.

%Finally, there is an important issue that the computational expense of the smallest eigenvector is too large. In order to get the smallest eigenvalue and eigenvector of the grounded Laplacian matrix $\LL_{-S}$ with a low computational cost, we introduce the method mentioned in~\cite{BaSpSrTe13,SpTe14,CoKyMiPaJaPeRaXu14}, and the performance of this method is given in the following lemma.

%First, we show the calculation of the smallest eigen-pair of the grounded Laplacian matrix $\LL_{-S}$ with time complexity of $\tilde{O}(m)$.

\begin{lem}\label{LapSolver}
Given a grounded Laplacian matrix $\LL_{-S}$ with smallest eigenvalue $\lambda(\LL_{-S})$, there exists a nearly-linear time algorithm $\uu=\textsc{Vector}(\LL_{-S},\epsilon)$ with the output  vector $\uu$ and $\uu^\top \LL_{-S} \uu$  satisfying $\uu^\top \LL_{-S} \uu \leq (1+\epsilon)\lambda(\LL_{-S})$. The algorithm takes $\tilde{O}(m)$ time, where $\tilde{O}(\cdot)$ hides the ${\rm poly} (\log n)$ factor.
\end{lem}

Note that except for the  approximation eigenvector $\uu$ of matrix $\LL_{-S}$, the algorithm $\textsc{Vector}(\LL_{-S},\epsilon)$ in Lemma~\ref{LapSolver} can be also used to approximate $\lambda(\LL_{-S})$ by $\uu^\top \LL_{-S} \uu$.

Next we are in a position to propose a fast approximation algorithm for Problem~\ref{prob:2}, the pseudo-code of which is outlined in Algorithm~\ref{alg:2}, and the efficiency of which is stated in Theorem~\ref{thm:time2}.

\begin{algorithm}[htbp]
	\caption{\textsc{Fast}$(\calG, k,S,\epsilon)$}
		\label{alg:2}
		\Input{
			A graph $\calG=(V,E)$; an integer $k$; a grounded node set $S$; an error parameter $\epsilon$
		}
		\Output{
			$T$: an edge set with $|T| = k$
		}
		Initialize solution $T = \emptyset$ \;
        Initialize candidate edge set $Q= (S \times (V\backslash S))\backslash E$ \;
	
		\For{$j = 1$ to $k$}{
            $\uu=\textsc{Vector}(\LL_{-S}(T),\epsilon)$\;
           %Compute the smallest eigenvector $\uu$ corresponding to the grounded Laplacian matrix $\LL_{-S}(T)$\;
           \For{each $e \in Q$}{
                $i=$ the non-grounded node $e$ is incident to\;
                Compute $\tilde{\lambda}_T(e)=2\uu_i \sum_{t \in \calN_i \backslash S}\uu_t$\;
           }
			Select $e$ s.t.  $e \gets \mathrm{arg\, max}_{e \in Q} \tilde{\lambda}_T(e)$ \;
            Update $Q \gets Q \backslash \{e\}$\;
			Update solution $T \gets T \cup \{e\}$ \;
		}
		\Return $T$
\end{algorithm}

\begin{thm}\label{thm:time2}
The time complexity of Algorithm~\ref{alg:2} is $\tilde{O}(km)$.
\end{thm}
\begin{proof}
In Algorithm~\ref{alg:2}, we first initialize the edge set $T$ and the candidate set $Q$ with size $O(n)$. Then in $k$ iterations, we  calculate the  eigenvector $\uu$ associated with the smallest eigenvalue of matrix $\LL_{-S}(T)$, which takes $\tilde{O}(m)$ time. Next, for all edges $e\in Q$, we calculate $\tilde{\lambda}_T(e)$, which takes $O(m)$ time. Finally, we select the edge that maximizes $\tilde{\lambda}_T(e)$ and add it to the edge set $T$, which need $O(n)$ time. Thus, the overall running time of Algorithm~\ref{alg:2} is $\tilde{O}(km)$.
\end{proof}

\begin{remark}
Theorem~\ref{thm:time2} shows that the time complexity of Algorithm~\ref{alg:2} is $\tilde{O}(km)$, which is nearly linearly with respect to the number of edges. Thus, among graphs with the same number of nodes, Algorithm~\ref{alg:2} is more efficient on sparser graphs. Specifically, for trees---the sparest class of connected graphs, the complexity of Algorithm~\ref{alg:2} is $\tilde{O}(kn)$. As shown in the experiments Section, for those networks with similar number of nodes, Algorithm~\ref{alg:2} runs faster on sparse networks than on dense ones.
\end{remark}

%\textcolor{blue}{
%\begin{remark}
%Based on Theorem~\ref{thm:time2}, the time complexity of Algorithm~\ref{alg:2} is $\tilde{O}(km)$, meaning that the algorithm's complexity is nearly linearly related to the number of edges. Consequently, the algorithm runs more efficiently on sparser networks with the same number of nodes. Specifically, in the case of a tree, which is the sparest connected network, the algorithm's complexity is $\tilde{O}(kn)$. In the subsequent experimental results, we observe that for networks with similar node sizes, the more edges there are, the longer Algorithm~\ref{alg:2} takes to run.
%\end{remark}}

\iffalse
\begin{remark}
In Algorithm~\ref{alg:2}, the approximation $\tilde{\lambda}_T(e)$ of  the upper bound  $\bar{\lambda}_T(e)$ is used to estimate  the eigengap $\lambda_T(e)$. It can be seen in the Experiment Section, Algorithm~\ref{alg:2} consistently outperforms all the baseline methods over a diverse of real-world networks, indicating that $\bar{\lambda}_T(e)$ is an excellent predictor of $\lambda_T(e)$.
\end{remark}
\fi

%it can be observed from a large amount of experiments that the results returned by two algorithms are very close to each other, and far more better than other methods. is approximated by its upper bound,

\section{Experiments}

In this section, we perform  experiments on a large set of real-life networks to demonstrate the effectiveness, efficiency, and scalability of our algorithms.  %Algorithms~\ref{alg:1} and~\ref{alg:2}.

\begin{table}
	\centering
		\caption{Running time, smallest eigenvalues of the grounded Laplacian matrices and their ratio   for Algorithms~\ref{alg:1} and~\ref{alg:2} on real-world networks.}\label{tab:data}
		%\resizebox{\columnwidth}{!}{
\fontsize{8}{8}\selectfont			
\begin{tabular}{m{2.5cm}m{1.2cm}<{\centering}m{1.2cm}<{\centering}p{0.9cm}<{\centering}m{0.8cm}<{\centering}m{0.8cm}<{\centering}m{0.8cm}p{0.8cm}<{\raggedleft}}
				\toprule
				\multirow{3}*{ Networks} &\multirow{3}*{ Nodes} &\multirow{3}*{ Edges} &   \multicolumn{2}{c}{Running} &  \multicolumn{3}{c}{Smallest Eigenvalue} \\
				& & & \multicolumn{2}{c}{Time (s)} &\multicolumn{3}{c}{$(\times 10^{-4})$}
				\\
				\cmidrule{4-5}
				\cmidrule{6-8}
				&  & & $\textsc{Greedy}$ & $\textsc{Fast}$ &  $\textsc{Greedy}$ & $\textsc{Fast}$ & Ratio \\
				\midrule
%                Tribes & 16& 58& &0.015 &0.001 &15& & 2889& 2882&1.002\\
 %               FirmHiTech & 33& 147& &0.065& 0.002& 33& & 1388& 1367&1.015\\
%                Karate & 34& 78& & 0.066& 0.002& 33&& 1252& 1221&1.025\\
%                Dolphins & 62& 159& &4.88& 0.274&18&  & 715.6& 705.5&1.014\\

                Minnesota & 2642& 3303&  9931& 2.63&   105& 99.4& 1.064\\
                USGrid & 4941& 6594&  20283& 3.88& 68.7& 66.8&1.028\\
                Bcspwr10 & 5300& 8271&  27127& 5.01&  69.1& 67.3&1.026\\
                RealityCall & 6809& 51247&  37219& 3.21&  74.8& 71.9&1.040\\
                Eurqsa & 7245& 20509&  61009& 7.49&  73.5& 69.3&1.062\\
                WHOIS & 7476& 56943&  67233& 8.04&  66.1& 65.8&1.003\\
                Rajat06 & 10922& 18061&  --& 11.8&  --& 27.9&--\\
                Indochina2004 & 11358& 47606&  --& 10.5& --& 40.3&--\\

                Amazon & 91813& 125704&  --& 179&--&   0.23&--\\
                Luxembourg& 114599& 119666&  --& 151&  --& 0.25&--\\
                Sk2005 & 121422& 334419&  --& 740&  --& 2.58&--\\
                Caida & 190914& 607610&  --& 294& --& 2.52&--\\
                Pwtk & 217891& 5653221&  --& 1278&  --& 2.03&--\\
                DBLP2010 & 226413& 716460&  --& 398&  --& 2.17&--\\
                TwitterFollows & 404719& 713319&  --&270& --& 1.22&--\\
                Delicious & 536108 & 1365961 &  --& 729&--&0.93&--\\
                FourSquare & 639014 & 3214986 &  --& 873&--&0.78&--\\
                BerkStan & 685230& 7600595&  --& 5138& --& 0.67&--\\
                RoadNetPA & 1088092& 3083796&  --& 3001&  --& 0.27&--\\
                YoutubeSnap & 1134890& 2987624&  --&2757&  --& 0.44&--\\
				\bottomrule
			\end{tabular}
		%}
\end{table}

\subsection{Setup}

\textbf{Datasets}. We select a large number of real-world networks of different sizes, up to one million nodes, which are publicly available in KONECT~\cite{Ku13} and SNAP~\cite{LeSo16}. Information of networks and their nodes and edges are shown in the first three columns of Table~\ref{tab:data}. All our experiments are conducted on the largest connected component for each of these networks.

%Due to the need of the problem, we complete the subsequent experiments on their maximum connected subgraphs.
 %Our algorithms are tested on a diverse set of real-world networks with up to one million nodes, all of which are publicly available in KONECT~\cite{Ku13} and SNAP~\cite{LeSo16}. Without loss of generosity, we only consider  connected networks, while for any disconnected network we  execute experiments on its largest component.  Relevant statistics of these datasets is summarized in Table~\ref{tab:data}, where networks are shown in increasing order of the numbers of nodes.

\textbf{Machine Configuration and Reproducibility}. Our algorithms are programmed and implemented in \textit{Julia}, for convenient use of algorithm $\textsc{Vector}(\LL_{-S},\epsilon)$ in the \textit{Julia} Laplacian.jl package, which is available on \url{https://github.com/danspielman/Laplacians.jl}. The error parameter $\epsilon$ is set to be $10^{-3}$ in the experiments. All our experiments are conducted on a machine equipped with 3.5 GHz Intel i5-4690K CPU and 16G memory, using a single thread. Our source code is available on \url{https://github.com/kedges/kedges}.

\textbf{Baseline Methods.} We select several baselines to compare with our proposed algorithms \textsc{Greedy} and \textsc{Fast}. The baseline methods are as follows.
\begin{enumerate}
    \item \textsc{Optimum}: choose $k$ edges from the candidate set $Q$, which maximize $\lambda(T)$ by exhaustive search.
	\item \textsc{Degree}: choose $k$ edges incident to  $k$ non-grounded nodes with largest degrees.
	\item \textsc{Eigenvector}~\cite{Ru00}: choose $k$ edges incident to  $k$ non-grounded nodes with highest eigenvector centrality associated with leading eigenvalue of adjacency matrix.
	\item \textsc{Betweenness}~\cite{Br01}: choose $k$ edges incident to $k$ non-grounded nodes with highest betweenness centrality. %BrPi07
	\item \textsc{Closeness}~\cite{Be65}: choose $k$ edges incident to  $k$ non-grounded nodes with highest closeness centrality. %BrPi07
    \item \textsc{$k$-center}~\cite{ShSoSaJo14,Go85}: choose one edge at one time to minimize the maximum distance between grounded nodes and non-grounded nodes by iterating $k$ times.
    \item \textsc{Eigen-approx}: select $k$ edges in an iterative way. At each round, choose one edge  incident to a non-grounded node with the largest component in the eigenvector $\uu$ associated with the smallest eigenvalue for the grounded Laplacian matrix.
\end{enumerate}

\subsection{Effectiveness and accuracy}

To exhibit the effectiveness of our two algorithms, we compare the results returned by our algorithms with the optimal solution on four small networks from KONECT~\cite{Ku13}: Dolphins with 62 nodes and 159 edges, Tribes with 16 nodes and 58 edges, Karate with 34 nodes and 78 edges, and FirmHiTech with 33 nodes and 147 edges. For each of these four networks, we randomly select one node as the grounded node, and then add  $k= 1, 2, \cdots, 5$ edges. Due to the small scales of these networks, we can get the optimum solution within an acceptable  time. We present the results in Fig.~\ref{fig:1}, which shows that both algorithms \textsc{Greedy} and \textsc{Fast} return very accurate results, which are very close to the optimum solution. %, and are far more better than the approximation ratio analysed in Theorem~\ref{them:appx}.

\begin{figure}[tbp]
  \centering
  % Requires \usepackage{graphicx}
  \includegraphics[width=1\linewidth]{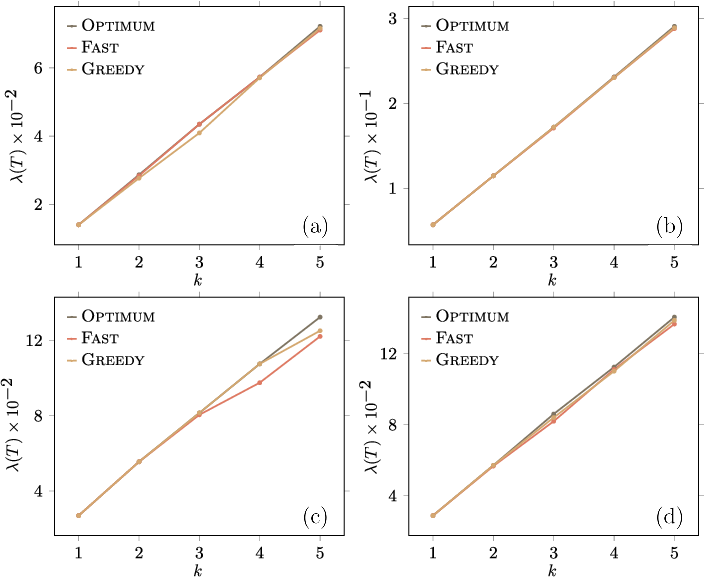}\\
\caption{$\lambda(T)$ given by \textsc{Fast}, \textsc{Greedy}, and \textsc{Optimum} schemes for $k$ ranging from 1 to 5 on four small networks: (a) Dolphins, (b) Tribes, (c) Karate and (d) FirmHiTech.}
\label{fig:1}
\end{figure}

To further show the similarity of algorithms \textsc{Greedy} and \textsc{Fast} in terms of effectiveness, we conduct experiments on six mid-sized networks  in the first six rows of Table~\ref{tab:data}. For each of these six networks, we randomly select 5 nodes as grounded nodes and set $k=50$, then compute the smallest eigenvalue of the grounded Laplacian matrix  after edge addition, by using algorithms \textsc{Greedy} and \textsc{Fast}. In the last three columns of Table~\ref{tab:data}, we can see the smallest eigenvalues returned by the two algorithms have little difference, since the ratio of them is close to 1.

%%%%%%%%%%%%%%%%%%%%%%%%%%%%%%%%%%%%%%%%%
\begin{figure}[tbp]
  \centering
  % Requires \usepackage{graphicx}
  \includegraphics[width=1\linewidth]{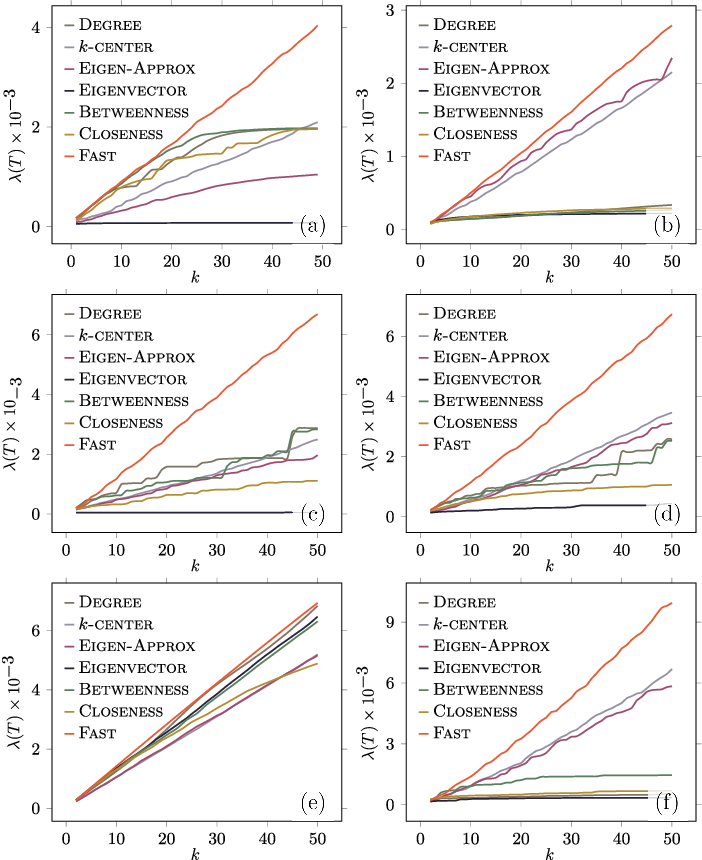}\\
\caption{$\lambda(T)$ returned  by \textsc{Fast},  \textsc{Degree},  \textsc{Eigenvector}, \textsc{Betweenness}, \textsc{Closeness}, \textsc{$k$-center} and \textsc{Eigen-Approx}, for $k$ ranging from 1 to 50 on six medium-size networks: (a) Indochina2004, (b) Rajat06, (c) USGrid, (d) Bcspwr10, (e) Eurqsa, and (f) Minnesota. }
\label{fig:3}
\end{figure}
%%%%%%%%%%%%%%%%%%%%%%%%%%%%%%%%%%%%%%%%%

%%%%%%%%%%%%%%%%%%%%%%%%%%%%%%%%%%%%
%\input{graph4.tex}
\begin{figure}[tbp]
  \centering
  % Requires \usepackage{graphicx}
  \includegraphics[width=1\linewidth]{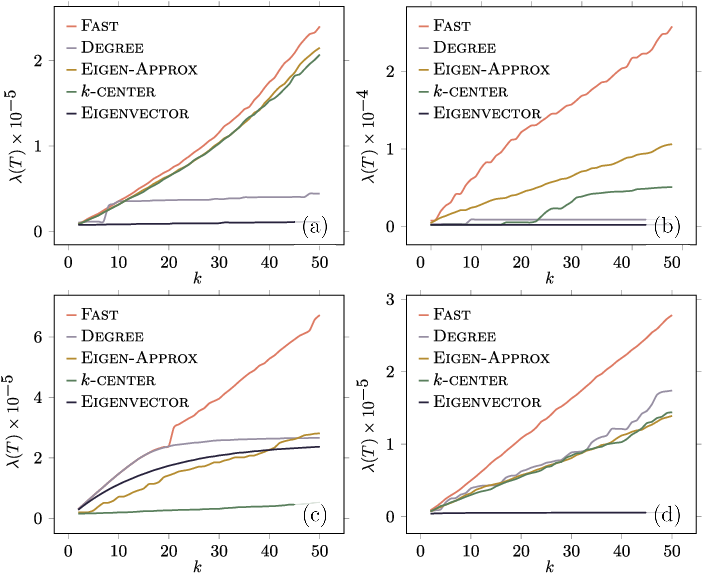}\\
\caption{$\lambda(T)$ returned  by \textsc{Fast} compared with  \textsc{Degree}, \textsc{$k$-center}, \textsc{Eigen-Approx} and \textsc{Eigenvector} with $k$ ranging from 1 to 50 on four large  networks: (a) Amazon, (b) Sk2005, (c) BerkStan, (d) RoadNetPA. }
\label{fig:4}
\end{figure}
%%%%%%%%%%%%%%%%%%%%%%%%%%%%%%%%%%%%

We also compare the results returned by algorithm \textsc{Fast} with those corresponding to six baseline methods on six mid-size real networks, including \textsc{Degree},  \textsc{Eigenvector}, \textsc{Betweenness}, \textsc{Closeness}, \textsc{$k$-center} and \textsc{Eigen-Approx}.  For each network, we randomly select 5 nodes as grounded nodes and add $k=1,2, \cdots, 50$ edges. The smallest eigenvalues returned by these 7 methods are reported in Fig.~\ref{fig:3}, which demonstrates that for Problem~\ref{prob:2}, algorithm \textsc{Fast} always  outperforms these six baselines.

\subsection{Efficiency and scalability}

Theorems~\ref{them:appx} and~\ref{thm:time2} show theoretically that  the  time complexity of algorithms \textsc{Greedy} and \textsc{Fast} differ greatly. Here we  show experimentally that algorithm \textsc{Fast} runs faster than algorithm \textsc{Greedy}. To this end, we randomly select 5 grounded nodes and add $k=50$ edges by our two algorithms for each network listed in Table~\ref{tab:data}, and report their running time for computing the  smallest eigenvalues of all augmented networks in the middle three columns of Table~\ref{tab:data}. For the networks with more than 10,000 nodes, algorithm \textsc{Greedy} fails due to its tremendous time cost that may be more than one day, while algorithm \textsc{Fast} finishes running within 2 hours even on networks with over one million nodes, indicating that \textsc{Fast} is scalable to large networks.

Besides our Algorithm \textsc{Fast}, the four baselines, \textsc{Degree}, \textsc{$k$-center}, \textsc{Eigen-Approx} and \textsc{Eigenvector}, can also quickly solve Problem~\ref{prob:2} on large networks, the time complexity of which are $O(n)$, $O(kn)$, $\tilde{O}(km)$ and $O(m)$, respectively. However, the effectiveness of algorithm \textsc{Fast} is much better than the four baselines.   In Fig.~\ref{fig:4}, we compare the results of algorithm \textsc{Fast} with those of the four baselines on four large networks with randomly selected 5 grounded nodes and varying number $k$ of adding edges. From Fig.~\ref{fig:4}, we observe that although our \textsc{Fast} is not the fastest, it outperforms the baseline methods.

\section{Related works}

In this section, we review some applications and properties of the smallest eigenvalue of the grounded Laplacian matrix, and the eigenvalue optimization problem by edge manipulation.

\textbf{Applications.} The notion of grounded Laplacian matrix $\LL_{-S}$ was first proposed in~\cite{Mi93}, which arises in various practical scenarios. In particular, its smallest eigenvalue $\lambda(\LL_{-S})$ contains a wealth of important information in different applications. For example, in leader-follower multi-agent systems describing opinion dynamics~\cite{BaHe06,PaBa10}, $\lambda(\LL_{-S})$ characterizes the convergence speed of the states for follower agents. While in pinning control of complex networks~\cite{WaCh02,LiWaCh04}, $\lambda(\LL_{-S})$ is used to measure the effectiveness of pinning control scheme~\cite{LiXuLuChZe21}. Other applications for grounded Laplacian matrix include vehicle platooning~\cite{BaJoMiPa12,HeMaHuSe15,PiBaJo22}, power systems~\cite{TeBaGa15}, centrality of a group of nodes~\cite{LiPeShYiZh19}, and so on.

%including leader-follower multi-agent systems~\cite{BaHe06,PaBa10}, pinning control of complex networks~\cite{WaCh02,LiWaCh04}, vehicle platooning~\cite{BaJoMiPa12,HeMaHuSe15}, and power systems~\cite{TeBaGa15}. It is now established that the spectrum of the grounded Laplacian matrix plays a fundamental role in characterizing the performance of these networked systems. For example, the smallest eigenvalue of $\LL_{-S}$ characterizes the convergence rate of leader-follower systems~\cite{RaJiMeEg09} and the effectiveness of pinning control scheme~\cite{LiXuLuChZe21}.

\textbf{Properties of the smallest eigenvalue $\lambda (\LL_{-S})$}. Due to the vast applications of the smallest eigenvalue $\lambda (\LL_{-S})$ for the grounded Laplacian matrix $\LL_{-S}$, it has attracted considerable attention of many groups in recent years, in order to analyze its properties and unveil its connection with other matrices. In~\cite{Pi14,PiSu14,PiSu16,PiShSu15}, upper and lower bounds of $\lambda (\LL_{-S})$ were provided; while in~\cite{LiXuLuChZe21}, some other properties of $\lambda (\LL_{-S})$ were analyzed in terms of the $(|S|+1)$-th smallest eigenvalue of the Laplacian matrix $\LL$ and the smallest degree of nodes in set $V \backslash S$. In addition, properties of $\lambda (\LL_{-S})$ for weighted undirected graphs~\cite{MaBe17} and directed graphs~\cite{XiCa17} were also studied.

%In~\cite{Pi14,PiSu14,PiSu16,PiShSu15}, upper and lower bounds on the smallest eigenvalue were provided, in terms of the sum of the weights of the edges between the grounded nodes in $S$ and the smallest eigenvector. In~\cite{PiShFiSu18}, graph–theoretic bounds were given for the smallest eigenvalue. While in~\cite{LiXuLuChZe21}, its properties were analyzed based on the $(s+1)$-th smallest eigenvalue of the Laplacian matrix $\LL$, the minimal degree of nodes in set $V \backslash S$, and the number of edges connected nodes in $S$ and $V \backslash S$. Properties for the smallest eigenvalue of grounded Laplacian matrix of weighted undirected~\cite{MaBe17} and directed~\cite{XiCa17} also received attention from the  scientific community.

\textbf{Eigenvalue optimization with edge manipulation.} Addition or deletion of edges is a commonly used approach in various practical optimization problems, including the optimization problems for specific eigenvalues or their ratios of different matrices. In~\cite{ChToPrElFaFa16,ChPeYiTo21} maximizing or minimizing the leading eigenvalue of the adjacency matrix was explored by adding or removing a fixed number of edges, which was shown to be NP-hard~\cite{VaStKuLiVaLiWa11}. In~\cite{ZhZhCh21}, the problem of minimizing the largest eigenvalue of non-backtracking matrix~\cite{LiChZh17,LiZh19} was addressed by removing a given number of edges. In~\cite{GhBo06}, the authors addressed the problem of maximizing the algebraic connectivity by creating a given number of edges, which is the smallest non-zero eigenvalue of the Laplacian matrix $\LL$. Finally, in~\cite{JaVeToLiAb21}, optimizing the synchronizability was considered by adding edges, where synchronizability is measured by the ratio of the largest eigenvalue and the smallest non-zero eigenvalue of Laplacian matrix $\LL$. Thus far, optimizing $\lambda (\LL_{-S})$ of grounded Laplacian matrix by adding edge has not been studied. Moreover, prior methods for optimizing eigenvalues of other matrices do not apply the problem studied in this paper.

\section{Conclusion}
In this paper, we proposed and studied the problem of maximizing the smallest eigenvalue of the grounded Laplacian matrix $\LL_{-S}$ of a graph $\calG=(V,E)$ with $n$ nodes, $m\ll n^2$ edges, and a grounded node set $S$, by adding $k$ nonexistent edges from the candidate edge set $Q=(V\times V)\backslash E$ to graph $\calG$. We proved that the objective function is monotonic increasing, but not submodular. Since the size of the candidate edge set is $O(n^2)$, we simplified the problem by restricting the candidate edge set to be its proper set $(S \times (V\backslash S))\backslash E$ with $O(n)$ edges, and proved that both the simplified problem and the original problem have identical optimal solutions. We then developed two greedy algorithms to solve the problem. The first is a simple greedy algorithm, with $O(kn^4)$ time complexity and a $(1-e^{-\gamma\alpha})/\alpha$ approximation ratio, where $\gamma$ and $\alpha$ are, respectively, the submodularity ratio and curvature of the objective function.  We provided a lower bound for $\gamma$ and an upper bound for $\alpha$, when subgraph $\calG[V\backslash S]$ is connected. The second one is a fast algorithm, with running time $\tilde{O}(km)$, where $\tilde{O}(\cdot)$ suppresses the ${\rm poly} (\log n)$ factors. Extensive experiments on real networks show that the results of our algorithms are close to each other, and both algorithms are more effective than baseline methods. Moreover, our fast algorithm is scalable to large networks with more than one million nodes. Finally, it should be mentioned that our methods and results only apply to undirected graphs. In future, we will extend or modify our techniques to directed graphs.

\section*{Declaration of competing interest}
The authors declare that they have no known competing financial interests or personal relationships that could have appeared to influence the work reported in this paper.

\section*{Acknowledgements}
This work was supported by the National Natural Science Foundation of China (No. U20B2051), the Shanghai Municipal Science and Technology Major Project (No.  2021SHZDZX03), and Ji Hua Laboratory, Foshan, China (No.X190011TB190).

%smallest eigenvalue maximization problem by adding $k<<n$ new edges to the network with a given leader set $S$. The convergence rate equals to the smallest eigenvalue of the grounded Laplacian matrix in the continuous DeGroot model. We discuss the monotonicity of the objective function when the new edge links two followers or one follower and one leader. And we give a heuristic method to show that adding edges between one follower and one leader is more effective than adding edges between two followers, which can help us greatly reduce the size of candidate edge set from $O(n^2)$ to $O(n)$. Under this new restriction, we show the non-submodularity of the objective function, but we give bounds of its submodularity ratio and curvature when the remaining network by deleting node set $S$ is connected, and they ensure an approximation ratio. We then design an algorithm with time complexity of $\tilde{O}(km)$ by maximizing $\lambda(T)$ in a iterative way under an approximation. Extensive experiments shows the effectiveness and efficiency of our fast algorithm, and it is scalable on networks even with more than one million nodes and outperforms other baselines.

\bibliographystyle{elsarticle-num}
%\bibliographystyle{ACM-Reference-Format}

%\bibliography{reference,newref,kedges,EigenOpt}
%\balance
\bibliography{newref}
%\biboptions{longnamesfirst}
%\biboptions{longnamesfirst,angle,semicolon}

\end{document}